\documentclass[11pt,titlepage]{article}
\usepackage[dvipsnames]{xcolor}
\usepackage{amsthm}
\usepackage{diagbox}
\usepackage{multirow}
\usepackage{longtable}
\newtheorem{lemma}{Lemma}

\newtheorem{definition}{Definition}
\newtheorem*{assumption}{Assumption}
\newtheorem{theorem}{Theorem}

\newtheorem{algo}{Algorithm}
\newtheorem{remark}{Remark}

\usepackage[latin1]{inputenc}
\usepackage{rotating, rotfloat}
\usepackage[round]{natbib} 
\usepackage{graphics, graphicx, hyperref, url, amsmath, subfigure, amsthm, amsfonts, enumerate, epstopdf, caption, booktabs, bm}

\hypersetup{
    pdfstartview={FitH},    
    colorlinks=true,        
    linkcolor=blue,         
    citecolor=blue,         
    filecolor=blue,         
    urlcolor=blue           
}
\def\u{\textbf{\textit{u}}} \def\x{\textbf{\textit{x}}} \def\z{\textbf{\textit{z}}} \def\vv{\textbf{\textit{v}}}
\def\y{\textbf{\textit{y}}} \def\a{\textbf{\textit{a}}}  \def\c{\textbf{\textit{c}}}
  \def\0{\boldsymbol 0}       \def\w{\textbf{\textit{w}}}

\def\n{\textbf{\textit{n}}}

\newcommand{\bec}{\begin{center}}
\newcommand{\enc}{\end{center}}
\newcommand{\bee}{\begin{eqnarray*}}
\newcommand{\ene}{\end{eqnarray*}}
\newcommand{\beq}{\begin{equation}}
\newcommand{\eeq}{\end{equation}}

\usepackage[inline]{enumitem}
\usepackage[ruled,vlined]{algorithm2e}
\usepackage{pifont}
\usepackage{caption}
\usepackage{arydshln}

\newcommand{\X}{\mathcal X}                     
\newcommand{\Y}{\mathcal Y}                     
\newcommand{\Yout}{\mathcal{Y}_{\mathrm{out}}}  

\newcommand{\R}{\mathcal{R}}
\newcommand{\slab}[1]{\mathrm{Sl}\left(#1\right)}           
\newcommand{\cslab}[1]{\mathrm{cSl}\left(#1\right)}         
\newcommand{\dirs}{\mathcal T}          
\newcommand{\fdirs}{\mathcal U}         
\newcommand{\SHD}{\mathrm{sHD}}
\newcommand{\textSHD}{\textbf{sHD}}
\newcommand{\textSHDA}{\textbf{sHD}$_{\mathcal A}$}
\newcommand{\ccol}{purple}              
\newcommand{\stp}[1]{\textcolor{\ccol}{\textbf{Step~{#1}.}}} 
\newcommand{\stpc}[2]{\textcolor{\ccol}{\textbf{Step~{#1}: {#2}.}}} 
\newcommand{\stpf}[1]{\textcolor{\ccol}{\textbf{\hspace{0.5em}\ding{228} {#1}.}}}
\newcommand{\fl}[1]{\mathcal V_{#1}}    
\newcommand{\flc}{\mathcal V}           

\newcommand{\bmu}{\bm\mu}

\newcommand{\Sph}[1][d-1]{\mathcal{S}^{#1}}	

\newcommand{\tr}{^\top}                                     
\newcommand{\I}[1][d]{\mathbb I_{#1}}                       

\DeclareMathOperator{\linOp}{span}
\newcommand{\lin}[1]{\linOp\left(#1\right)}                 
\DeclareMathOperator{\intrOp}{int}			          		
\newcommand{\intr}[1]{\intrOp\left(#1\right)}
\DeclareMathOperator{\bdOp}{bd}			          		
\newcommand{\bd}[1]{\bdOp\left(#1\right)}

\newcommand{\pkg}[1]{{\normalfont\fontseries{b}\selectfont #1}}
\let\proglang=\textsf
\let\code=\texttt
\newcommand{\tru}{\code{TRUE}}
\newcommand{\fal}{\code{FALSE}}

\begin{document}

\title{\bf Exact and approximate computation of the scatter halfspace depth}
\vskip 5mm

\author {{Xiaohui Liu$^{a}$ ,
        Yuzi Liu$^{a}$,
        Petra Laketa$^{b}$,
        Stanislav Nagy$^{b,}$\footnote{Corresponding author's email: nagy@karlin.mff.cuni.cz.},
        Yuting Chen$^{a}$
        }\\ \\[1ex]
        {\em\footnotesize $^a$ School of Statistics, Jiangxi University of Finance and Economics, Nanchang, Jiangxi 330013, China}\\
        {\em\footnotesize $^b$ Faculty of Mathematics and Physics, Charles University, Prague, Czech Republic}\\
}

\maketitle

\textbf{Abstract}. The scatter halfspace depth (\textSHD{}) is an extension of the location half{\-}space (also called Tukey) depth that is applicable in the nonparametric analysis of scatter. Using \textSHD{}, it is possible to define minimax optimal robust scatter estimators for multivariate data. The problem of exact computation of \textSHD{} for data of dimension $d \geq 2$ has, however, not been addressed in the literature. We develop an exact algorithm for the computation of \textSHD{} in any dimension $d$ and implement it efficiently using \proglang{C++} for $d \leq 5$, and in \proglang{R} for any dimension $d \geq 1$. Since the exact computation of \textSHD{} is slow especially for higher dimensions, we also propose two fast approximate algorithms. All our programs are freely available in the \proglang{R} package \code{scatterdepth}.
\vspace{2mm}

\vspace{2mm}


\setlength{\baselineskip}{1.5\baselineskip}

\vskip 0.1 in
\section{Introduction: Scatter halfspace depth and its computation}
\paragraph{}
\vskip 0.1 in \label{Introduction}

The location halfspace (or Tukey) depth is a powerful method of nonparametric analysis that is applicable to multivariate data \citep{Tukey1975, Donoho_Gasko1992, Chen_etal2018}. For a dataset $\X$ composed of observations $\x_1, \dots, \x_n \in \R^d$, the location halfspace depth of a point $\x \in \R^d$ with respect to (w.r.t.) the dataset $\X$ can be defined as the smallest number of data points lying in a closed halfspace whose boundary hyperplane passes through $\x$
    \begin{equation}\label{equation:halfspace depth} 
    \mathrm{HD}(\x, \X) = \inf_{\u \in \Sph} \# \left\{ i \in \left\{ 1, \dots, n \right\} \colon \u\tr \x_i \le \u\tr \x \right\}.
    \end{equation}
Here, $\Sph = \left\{ \u \in \R^d \colon \left\Vert \u \right\Vert = 1 \right\}$ is the unit sphere in $\R^d$, and $\# A$ stands for the number of elements of a finite set $A$. The unit vector $\u \in \Sph$ plays in~\eqref{equation:halfspace depth} the role of the outer normal vector to the halfspace whose mass we evaluate.

The location halfspace depth~\eqref{equation:halfspace depth} provides a natural data-dependent ordering in $\R^d$ --- points $\x \in \R^d$ positioned ``centrally'' w.r.t. the data cloud $\X$ obtain higher values of the depth~\eqref{equation:halfspace depth}, while peripheral points are distinguished by their low depth values. This ordering allows us to devise a theory of nonparametric statistics based on ranks in $\R^d$. A prime example is (the barycenter of) the set of maximizers of the depth function~\eqref{equation:halfspace depth} over $\x \in \R^d$. That collection of deepest points is frequently called the set of halfspace medians of $\X$, and presents a plausible analogue of the univariate median suitable for $\R^d$-valued observations \citep{Donoho_Gasko1992, Chen_etal2018}. Another prominent example of a depth-based method is the bagplot, which is a version of the boxplot applicable to multivariate data \citep{Rousseeuw_etal1999}. All these concepts are already firmly embedded in the nonparametrics of multivariate data, and a substantial progress has been made into their theory, practice, and computational aspects.

The location halfspace depth introduces elements of the first-order inference to $\R^d$; in \eqref{equation:halfspace depth} we indeed evaluate points $\x \in \R^d$ for their suitability as location estimators of $\X$. A next step is the analysis of the dispersion of $\X$. Expanding the scope of the location halfspace depth toward the scatter of multivariate data clouds, the scatter halfspace depth (\textSHD{}) has been proposed. \textSHD{} is based on the ideas of \cite{Zhang2002}, and was formally introduced by \cite{Chen_etal2018} and \cite{Paindaveine_VanBever2018}. For a given dataset $\X$ of $\R^d$-valued observations $\x_1, \dots, \x_n$ and a given centering vector $\bmu \in \R^d$, the scatter halfspace depth\footnote{Note that similarly as for the location halfspace depth, the general definition of the scatter halfspace depth in \cite{Chen_etal2018} and \cite{Paindaveine_VanBever2018} is given, instead of only for datasets $\X$, for Borel probability measures $P$ on $\R^d$. In that case, the number of elements used in~\eqref{equation:scatter depth} is replaced by the $P$-mass of appropriate sets. Strictly speaking, our definition~\eqref{equation:scatter depth} should be divided by $n$ to match the sample depth from the literature exactly.  It will, however, be convenient for us to consider the integer-valued version of the sample \textSHD{} as defined in~\eqref{equation:scatter depth}.The two versions of the depth are, of course, equivalent.} evaluates the depth of a positive definite matrix $\Sigma \in \R^{d \times d}$ w.r.t. $\X$ as
    \begin{equation}
    \label{equation:scatter depth}
    \begin{aligned}
    \SHD( & \Sigma, \bmu, \X) \\
    & = \inf_{\u \in \Sph} \min\left\{\# \left\{ i \colon (|\u\tr (\x_i - \bmu)| \le \sqrt{\u\tr \Sigma \u} \right\}, \# \left\{ i \colon |\u\tr (\x_i - \bmu)| \ge \sqrt{\u\tr \Sigma \u} \right\} \right\}.
    \end{aligned}
    \end{equation}
Of course, the index $i$ in \eqref{equation:scatter depth} takes values in $\left\{1, \dots, n \right\}$. A typical choice of the centering vector $\bmu \in \R^d$ is the halfspace median of the dataset $\X$. Throughout this paper we, however, leave $\bmu$ unspecified, and treat \textSHD{} with any positive definite matrix $\Sigma$ and any location vector $\bmu$. 

Analogously to the location halfspace depth~\eqref{equation:halfspace depth}, the scatter halfspace depth~\eqref{equation:scatter depth} assigns higher values to matrices that fit into the dispersion pattern generated by the dataset $\X$. The median scatter matrix (called also the scatter halfspace median) of $\X$ can be defined as (a representative of) the set of positive definite matrices that maximize the function~\eqref{equation:scatter depth} with $\bmu$ given. The scatter halfspace median is known to possess several excellent theoretical properties; in particular, with $\bm\mu$ the halfspace median, it has been shown to be a minimax optimal estimator of scatter in certain contamination models \citep{Chen_etal2018}. 

The location halfspace depth~\eqref{equation:halfspace depth} is known to be difficult to compute exactly. Nevertheless, great progress has already been made in that direction, and efficient exact algorithms for the computation of the location halfspace depth~\eqref{equation:halfspace depth} are currently available. We refer to the pioneering works of \cite{Rousseeuw_Ruts1996} for bivariate data, \cite{Rousseeuw_Struyf1998} for 3-dimensional data, and more recently \cite{Dyckerhoff_Mozharovskyi2016} for general dimension $d \geq 1$. The algorithm for bivariate data is based on the general idea of a circular sequence \citep{Edelsbrunner1987}; the situation with $d > 2$ is much more involved and requires specialized techniques \citep{Liu_Zuo2014, Dyckerhoff_Mozharovskyi2016, Liu_etal2019}.

Despite its theoretical appeal, efficient computation of \textSHD{} for $d \geq 2$ is virtually unexplored. The only available implementations of \textSHD{} are based on random approximations of~\eqref{equation:scatter depth}, by means of a finite minimum over randomly chosen directions $\u \in \Sph$. Such approximations are well known from the task of computing the location halfspace depth \citep{Dyckerhoff_etal2021}, and are notorious for their poor performance in higher dimensions (\citealp[see also our Section~\ref{Sec:Illustrations}]{Nagy_etal2020}). Unlike for the location halfspace depth, no exact programs are available for the computation of \textSHD{}. 

We address the problem of both exact and approximate computation of \textSHD{} in arbitrary dimension $d \geq 1$. Based on thorough geometric considerations, we develop an original computational algorithm for \textSHD{}. Our algorithm provides exact results under a mild condition of sufficiently general position of the data points $\X$, which is satisfied if $\X$ is sampled from a distribution with a density. In addition, we propose two versions of fast approximate algorithms for the computation of \textSHD{}. For dimension $d \leq 5$ we have efficiently implemented our program in \proglang{C++} and \proglang{R} via the interface given by \code{Rcpp} \citep{Eddelbuettel2013} and \code{RcppArmadillo} \citep{Eddelbuettel_Sanderson2014}. All these programs are freely available in an \proglang{R} package \code{scatterdepth} accompanying this paper.\footnote{\url{https://github.com/NagyStanislav/scatterdepth}}

The rest of this paper is organized as follows. The methodology and the main theoretical results are presented in Section \ref{Sec:MMS}. Based on that discussion, we develop both exact and approximate computational algorithms in Section \ref{Sec:algorithms}, and comment on their implementation. Two numerical studies are given in Section \ref{Sec:Illustrations} to investigate the computational speed, and to compare the approximate versions of the developed implementations. An extensive proof of our main theoretical result and several proofs of auxiliary lemmata are given in the Appendix.

\paragraph{Notations.} We make frequent use of tools and terminology from geometry; our general reference is \cite{Schneider2014}. For a set $A \subseteq \R^d$ its interior and boundary are denoted by $\intr{A}$ and $\bd{A}$, respectively. The set $A$ is an \emph{affine (sub)space} (also called a \emph{flat}) of dimension $k \leq d$ if $A = \left\{ \x + \a \colon \x \in L \right\}$ for some point $\a \in \R^d$ and some $k$-dimensional linear subspace $L$ of $\R^d$. The \emph{affine hull} of a set $A$ is the smallest affine space that contains $A$. For example, the affine hull of two distinct points in $\R^d$ is the infinite line spanned by them. The elements of the unit sphere $\Sph = \left\{ \x \in \R^d \colon \left\Vert \x \right\Vert = 1 \right\}$ are called \emph{directions}. We say that a set $A \subset \R^d$ is a \emph{$k$-sphere} if $k \in \left\{ 0, \dots, d-1 \right\}$ and if $A$ can be written as an intersection of a sphere $\left\{ \x \in \R^d \colon \left\Vert \x - \a \right\Vert = r \right\}$ for some $\a \in \R^d$ and $r > 0$ with a $k$-dimensional affine subspace of $\R^d$. In particular, a $0$-sphere is any pair of points in $\R^d$.

\vskip 0.1 in
\section{Methodology and main results}
\paragraph{}
\vskip 0.1 in 
\label{Sec:MMS}

Our algorithm is based on considerations from spherical geometry, coupled with elements of convex geometry, and matrix theory. It is described in several steps. We begin in Section~\ref{section:scatter halfspace depth}, where we interpret \textSHD{} in geometric terms. Our algorithm is guaranteed to give an exact result under a mild assumption of sufficiently general position of the data points, introduced in Section~\ref{section:general position}. To simplify the exposition, in Section~\ref{section:affine invariance} we use the affine invariance of \textSHD{} to reduce the problem of the computation of \textSHD{}. Instead of computing the depth for any positive definite matrix $\Sigma \in \R^{d \times d}$ centered at an arbitrary location $\bmu\in\R^d$, we use the equivalent formulation of finding \textSHD{} of the unit $d \times d$ matrix $\I$ centered at the origin $\bm0_d$, w.r.t. a transformed sample of data points. Having covered the preliminaries, we state our main result as Theorem~\ref{theorem:main} in Section~\ref{section:main theorem}, and outline the way how to use that result to develop our algorithm in Section~\ref{section:idea of algorithm}.

\subsection{Geometry of the scatter halfspace depth} \label{section:scatter halfspace depth}

The scatter halfspace depth has a geometric interpretation \citep{Nagy2020} in terms of probabilities of the slabs
    \begin{equation}\label{equation:general slab}
    \slab{\u, \Sigma, \bmu} = \{\x \in \R^d \colon |\u\tr (\x - \bmu)| \le \sqrt{\u\tr \Sigma \u}\}.
    \end{equation}
The slab $\slab{\u, \Sigma, \bmu}$ is located between two hyperplanes that are both \begin{enumerate*}[label=(\roman*)] \item orthogonal to $\u\in \Sph$, and \item supporting the (boundary of the) Mahalanobis ellipsoid of $\Sigma$ centered at $\bmu$\end{enumerate*}
    \begin{equation}\label{equation:Mahalanobis}
    E = \{\x \in \R^d \colon (\x - \bmu)\tr \Sigma^{-1} (\x - \bmu) = 1\}.
    \end{equation}
Denoting the closure of the complement of $\slab{\u, \Sigma, \bmu}$ in $\R^d$ by
    \[
    \cslab{\u, \Sigma, \bmu} = \{\x \in \R^d \colon |\u\tr (\x - \bmu)| \ge \sqrt{\u\tr \Sigma \u}\},
    \]
we can rewrite~\eqref{equation:scatter depth} as
    \[  
    \SHD(\Sigma,\bmu, \X) = \inf_{\u \in \Sph} \min\{ \# \left\{ i \colon \x_i \in \slab{\u, \Sigma, \bmu} \right\}, \# \left\{ i \colon \x_i \in \cslab{\u, \Sigma, \bmu} \right\} \}.
    \]

\subsection{The assumption of general position w.r.t. ellipsoids} \label{section:general position}

To avoid clunky notation coupled with technical and numerical difficulties, throughout this paper we assume that the data $\X$ is in position that is sufficiently general. Recall that a (multi)set $\mathcal Z$ of points in $\R^d$ is said to be \emph{in general position} if no $k$-point subset of $\mathcal Z$ lies in a $(k-2)$-dimensional affine space, for any $k \in \left\{ 2, \dots, d+1 \right\}$. The assumption of data points lying in general position is quite standard in the literature on the computation of the (location) depth. 

For the scatter halfspace depth of a dataset $\X$ we need to assume a slightly different condition that we call general position w.r.t. the (Mahalanobis) ellipsoid $E$ from \eqref{equation:Mahalanobis}. To formulate our condition, denote for $\u \in \Sph$ by $H_{\u}$ the boundary hyperplane of the unique tangent halfspace\footnote{A halfspace $H$ is said to be tangent to $E$ if $H \cap E$ is a single point.} to $E$ with inner normal $\u$. 

\begin{definition}  \label{definition:general position}
Let $\Sigma \in \R^{d \times d}$ be positive definite and let the center $\bmu \in \R^d$ be given. For a dataset $\X$ of points $\x_1, \dots, \x_n \in \R^d$, consider the amended dataset $\X^*$ of $2\,n$ points consisting of the original data $\x_1, \dots, \x_n$ and the points $2\,\bmu-\x_1, \dots,2\,\bmu -\x_n$ reflected around the center $\bmu$. We say that $\X$ is \emph{in general position w.r.t. the ellipsoid~\eqref{equation:Mahalanobis}} if for each $\u \in \Sph$ the finite (multi)set $\left( H_{\u} \cap E \right) \cup \left\{ H_{\u} \cap \X^* \right\}$ of points inside the hyperplane $H_{\u}$ is in general position.
\end{definition}

\begin{assumption}
Throughout this paper we assume that the dataset $\X$ is in general position w.r.t. $E$ from~\eqref{equation:Mahalanobis}.
\end{assumption}

In Definition~\ref{definition:general position} we assume that the hyperplane $H_{\u}$ always contains at most $d-1$ points from $\X^*$, and when taken together with the unique point $H_{\u} \cap E$ on the boundary of $E$, these points are in general position inside $H_{\u}$. To phrase our condition in other words, we suppose that the union of the two hyperplanes $H_{\u}$ and $H_{-\u}$, which is the boundary of the slab $\slab{\u, \Sigma, \bmu}$ of $E$ from \eqref{equation:general slab}, contains at most $d-1$ data points, and if the points from the hyperplane $H_{-\u}$ are flipped around the center $\bmu$ of $E$ to $H_{\u}$, all the data points in $H_{\u}$ together with $H_{\u} \cap E$ are in general position. In particular, we assume that \begin{enumerate*}[label=(\roman*)] \item no data point lies exactly on the boundary of $E$, \item no two data points determine a straight line that touches $E$, etc. \end{enumerate*} Our Definition~\ref{definition:general position} is similar to the usual assumption of general position in $\R^d$, but it neither implies nor is implied by the latter. The assumption of general position w.r.t. $E$ is almost surely satisfied for random samples $\X$ from absolutely continuous distributions in $\R^d$.

\subsection{Affine invariance} \label{section:affine invariance}

The scatter halfspace depth is affine invariant \cite[Theorem~2.1]{Paindaveine_VanBever2018}, meaning that the task of computing $\SHD(\Sigma, \bmu, \X)$ can be reduced to finding the depth of the identity matrix $\I \in \R^{d \times d}$ w.r.t. an affine transform of the original dataset $\X$. Indeed, because $\Sigma \in \R^{d \times d}$ is positive definite, there exists a unique symmetric matrix $\Sigma^{-1/2} \in \R^{d \times d}$ with the property $\Sigma^{-1/2} \Sigma \Sigma^{-1/2} = \I$ \citep[Theorem~7.2.6]{Horn_Johnson2013}. Transforming $\y_i = \Sigma^{-1/2}\left( \x_i - \bmu\right)$, $i=1,\dots, n$, and denoting by $\Y$ the dataset of the transformed points $\y_1, \dots, \y_n$, we obtain $\SHD\left(\Sigma, \bmu, \X \right) = \SHD\left( \I, \bm 0_d, \Y \right)$. The corresponding Mahalanobis ellipsoid~\eqref{equation:Mahalanobis} reduces to the unit sphere $\Sph$, and the slab~\eqref{equation:general slab} and its complement simplify to
    \[
    \slab{\u, \I, \bm 0_d}  = \{\x \in \R^d \colon |\u\tr \x| \le 1\}, \quad \mbox{ and }\quad
    \cslab{\u, \I, \bm 0_d}  = \left\{ \x \in \R^d \colon | \u\tr \x | \geq 1 \right\}, 
    \]
with their two boundary hyperplanes at $\u$ and $-\u$
    \begin{equation}\label{equation:Hu} 
    H_{\u} = \left\{ \x \in \R^d \colon \x\tr\u = 1 \right\} \quad \mbox{ and }\quad  H_{-\u} = - H_{\u} = \left\{ \x \in \R^d \colon \x\tr\u = -1 \right\}
    \end{equation}
supporting the unit sphere $\Sph$. As a result, the computation of $\SHD(\I, \bm 0_d, \Y)$ is easier to describe, and throughout this section we therefore focus on the situation when $\Sigma = \I$, $\bmu = \bm 0_d$, and the dataset $\Y$. This allows us to simplify the notations and the geometric arguments that will follow. We write $\SHD(\Y)$ for $\SHD\left(\I, \bm 0_d, \Y\right)$, and  $\slab{\u}, \cslab{\u}$ instead of $\slab{\u, \I, \bm 0_d}, \cslab{\u, \I, \bm 0_d}$, respectively. Likewise, the assumption of general position w.r.t. $E$ from Definition~\ref{definition:general position} simplifies, as $E=\Sph$, and $H_{\u} \cap E  = \left\{\u\right\}$ for all $\u \in \Sph$. The general position of $\X$ w.r.t. $E$ from~\eqref{equation:Mahalanobis} is equivalent with the general position of $\Y$ w.r.t. $\Sph$.

Thanks to the affine invariance, the computation of the scatter halfspace depth reduces to the minimization of the objective function
    \begin{equation}\label{equation:objective function}
    \begin{aligned}
    h(\u) & = \min\left\{ \#\left\{ i \in \{ 1, \dots, n\} \colon \left\vert \u\tr \y_i \right\vert \leq 1 \right\}, \#\left\{ i \in \{ 1, \dots, n\} \colon \left\vert \u\tr \y_i \right\vert \geq 1 \right\} \right\} \\
    & = \min\left\{\#\left\{ i \in \{ 1, \dots, n\} \colon \y_i \in \slab{\u} \right\}, \#\left\{ i \in \{ 1, \dots, n\} \colon \y_i \in \cslab{\u} \right\}\right\}
    \end{aligned}
    \end{equation}
over all $\u \in \Sph$. The main idea of our computational algorithm is a significant reduction of the domain $\Sph$ of the function $h$ to a finite subset $\fdirs$ with a guarantee that the infimum of the function $h$ over $\Sph$ coincides with the minimum of (a slightly modified version of) $h$ over $\fdirs$.

\subsection{Main theorem} \label{section:main theorem}

\begin{definition}  \label{definition:MTH}
Denote by $\Y^*$ the dataset of $2\,n$ points $\y_1, \dots, \y_n, -\y_1, \dots, -\y_n$. We say that $H_{\u}$ from~\eqref{equation:Hu} is a \emph{maximal tangent hyperplane} if there is no direction $\vv\in\Sph$ such that $H_{\u}\cap \Y^*$ is a strict subset of $H_{\vv}\cap \Y^*$. The direction $\u \in \Sph$ corresponding to a maximal tangent hyperplane is called a \emph{maximal direction}, and the finite set $\Y^* \cap H_{\u}$ is called a \emph{maximal subset} of $\Y^*$.
\end{definition}

To phrase Definition~\ref{definition:MTH} differently, a maximal tangent hyperplane $H_{\u}$ cannot be modified in a way to obtain a different tangent hyperplane of $\Sph$ that would contain the same points as $H_{\u} \cap \Y^*$, and some additional points from $\Y^*$. We are now ready to state our main theorem.

\begin{theorem} \label{theorem:main}
Let $\Y$ be a dataset in general position w.r.t. the unit sphere $\Sph$. Let $\fdirs \subset \Sph$ be a finite collection of directions such that 
    \begin{itemize}
        \item for each maximal subset $\left\{ \a_1, \dots, \a_{d-1} \right\}$ of $\Y^*$ of $d-1$ points, both corresponding maximal directions are contained in $\fdirs$, and
        \item for each maximal subset $\left\{ \a_1, \dots, \a_{k} \right\}$ of $\Y^*$ of $k < d-1$ points, any single direction from the corresponding set of maximal directions is contained in $\fdirs$.
    \end{itemize}
Then we can write
    \begin{equation}\label{equation:main theorem}
    \SHD(\Y) = \min_{\u \in \fdirs} \min\left\{ \#\left\{ i \in \{ 1, \dots, n\} \colon \left\vert \u \tr \y_i \right\vert < 1 \right\}, \#\left\{ i \in \{ 1, \dots, n\} \colon \left\vert \u \tr \y_i \right\vert > 1 \right\} \right\}.  
    \end{equation}
\end{theorem}

The proof of Theorem~\ref{theorem:main} is technical and long; in full detail it is contained in the Appendix, Section~\ref{section:proof main}. 

\begin{remark}  \label{remark:symmetry}
Since the objective function to be minimized in~\eqref{equation:main theorem} is symmetric w.r.t. antipodal reflections of directions $\Sph \to \Sph \colon \u \mapsto -\u$, it is possible to restrict the set of directions $\fdirs$ to those contained in an arbitrary hemisphere of $\Sph$ in~\eqref{equation:main theorem}. This argument of symmetry will be used to simplify our main algorithm in Section~\ref{Sec:algorithms}.
\end{remark}

There are finitely many directions in the set $\fdirs$ in Theorem~\ref{theorem:main}. They can be found using the notion of circles and anti-circles of data points from $\Y$, dealt with in detail in Section~\ref{section:shells} in the Appendix. Let $\y_i \in \Y$ be a data point such that $\left\Vert \y_i \right\Vert > 1$. The \emph{circle} of $\y_i$ is the set 
    \begin{equation}\label{equation:circle}
    \Sph \cap \left\{ \x \in \R^d \colon \x\tr \y_i = 1 \right\}.
    \end{equation}
The \emph{anti-circle} of $\y_i$ is the circle of its antipodal reflection $-\y_i \in \Y^*$. Each (anti-)circle,\footnote{By (anti-)circle we mean \emph{circle and/or anti-circle}, with precise meaning obvious from the context.} being an intersection of the unit sphere $\Sph$ and a hyperplane, is a $(d-2)$-sphere lying inside the unit sphere $\Sph$, see Figures~\ref{fig:d2} and~\ref{fig:d3}. The circle is a $1$-sphere (or a circle in the sense of a simple closed curve) on the surface of $\Sph[2]$ only for $d=3$; in the sequel it will, however, be convenient to use this terminology.\footnote{To avoid ambiguous terminology, we reserve the term \emph{circle} only to this context, and call $\Sph[1]$ consistently the (unit) \emph{$1$-sphere} throughout this paper.} The connections between (anti-)circles of data points and maximal tangent hyperplanes are expounded in Section~\ref{section:maximal tangent hyperplanes} in the Appendix. In brief, a maximal direction is any direction that lies in an intersection $C$ of $k \in \left\{ 1, \dots, d-1 \right\}$ (anti-)circles that have the property that no $(k+1)$-st (anti-)circle intersects $C$. In dimension $d=2$, the maximal directions are exactly the points of (any) (anti-)circles. For dimensions $d=3$ and $d=4$ see Figures~\ref{fig:d3} and~\ref{fig:d4}, respectively.

\subsection{The main algorithm: General description}    \label{section:idea of algorithm}

Based on Theorem~\ref{theorem:main}, we design our main algorithm for the exact computation of the scatter halfspace depth.

\renewcommand{\algorithmcfname}{Main algorithm}
\renewcommand{\thealgocf}{}
\begin{algorithm}
	\caption{The procedure for finding the exact scatter halfspace depth \textSHD{} of $\Sigma \in \R^{d \times d}$ positive definite, centered at $\bmu \in \R^d$, w.r.t. $\x_1, \dots, \x_n \in \R^d$.}
	\begin{enumerate}[label=\textbf{(S$_{\arabic*}$)}, ref=\upshape{\textbf{(S$_{\arabic*}$)}}]
	    \item \label{S1} Transform $\y_i = \Sigma^{-1/2}(\x_i - \bmu)$, $i=1,\dots,n$, and take the amended dataset $\Y^*$ consisting of $\y_1, \dots, \y_n, -\y_1, \dots, -\y_n$. 
		\item \label{S2} Initialize the list of maximal directions $\fdirs=\emptyset$.
		\item \label{S3} Loop through all $k$-point subsets $\{\a_1,\dots,\a_k\}$ of $\Y^*$ such that $\a_{j_1} \ne \pm \a_{j_2}$ for all $j_1 \ne j_2$, $j_1, j_2 \in \left\{ 1, \dots, k \right\}$, and all $k \in \left\{ 1, \dots, d-1 \right\}$. That is
		
		\textbf{For each} $k \in \left\{ 1, \dots, d-1 \right\}$ and \textbf{for each} $\{\a_1,\dots,\a_k\} \subset \Y^*$ as above \textbf{do}
		
		\begin{itemize}
			\item[\ding{228}] Search for any maximal tangent hyperplane $H_{\u}$ that satisfies $\{\a_1,\dots, \a_k\}=H_{\u}\cap \Y^*$.
			\item[\ding{228}] \textbf{If} such a maximal tangent hyperplane $H_{\u}$ does not exist, \textbf{then} break out of the current loop and continue with the next subset in Step~\ref{S3}.
			\item[\ding{228}] \textbf{Else}, we distinguish two cases:
		        \begin{itemize}
		            \item[\ding{222}] \textbf{If} $k=d-1$, \textbf{then} we find both directions $\u$ and $\u'$ such that $H_{\u} \cap \Y^* = H_{\u'} \cap \Y^* = \{\a_1,\dots,\a_{d-1}\}$, and add both $\u$ and $\u'$ to $\fdirs$.
		            \item[\ding{222}] \textbf{If} $k<d-1$, \textbf{then} we take an arbitrary single direction $\u\in\Sph$ that satisfies $H_{\u}\cap \Y^*=\{\a_1,\dots,\a_k\}$ and add $\u$ to $\fdirs$.
		        \end{itemize}
		\end{itemize}
		\item \label{S4} \textbf{Return} \textSHD{} as~\eqref{equation:main theorem} from Theorem~\ref{theorem:main} with the given set $\fdirs$.
	\end{enumerate}
\end{algorithm}

In dimension $d=2$ our algorithm simplifies since in Step~\ref{S3} the case $k < d-1 = 1$ is never considered. The algorithm thus in the main loop in Step~\ref{S3} reduces to the search of two tangent hyperplanes (lines) to $\Sph[1]$ from all points of the amended dataset $\Y^*$, see also Figure~\ref{fig:d2}. 

For $d=3$, the algorithm searches through all pairs of points in Step~\ref{S3} with $k = d-1 = 2$ and their maximal directions, and then with $k = d-2 = 1$ it considers also data points from $\Y^*$ whose (anti-)circles do not intersect any other (anti-)circle. This situation is displayed in Figures~\ref{fig:d3} and~\ref{fig:circles}.

With $d>3$ the procedure is substantially more involved, as all cases of $k \in \left\{ 1, \dots, d-1 \right\}$ (anti-)circles must be treated in Step~\ref{S3}. For a hint with $d=4$ see Figure~\ref{fig:d4}. The technical implementation of our procedure and details of the computation are provided in Section~\ref{Sec:algorithms}. 

\begin{figure}[htpb]
    \centering
    \includegraphics[width=.85\textwidth]{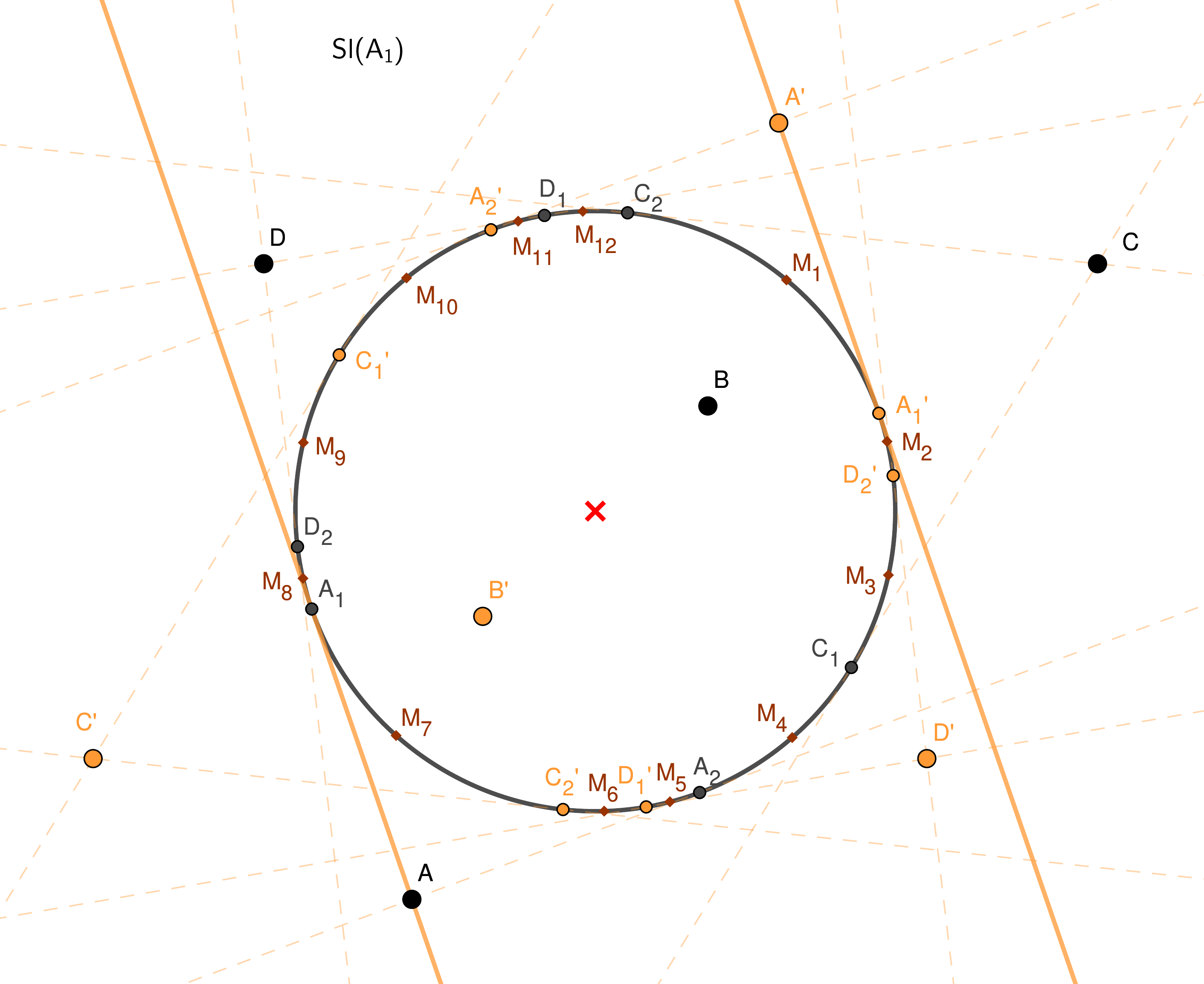}
    \caption{The situation in dimension $d=2$: 
    A dataset $\Y$ of $n = 4$ points (black points labeled $A$--$D$) and their antipodal reflections (orange points labeled $A'$--$D'$). The circles of $\Y$ are displayed as the points on $\Sph[1]$ in black ($A_1, A_2, C_1, C_2, D_1, D_2$); the anti-circles of $\Y$ are the points on $\Sph[1]$ in orange ($A_1', A_2', C_1', C_2', D_1', D_2'$). According to Theorem~\ref{theorem:main}, the depth \textSHD{} can be computed using slabs $\slab{\u}$ with $\u$ taken from the (anti-)circles of $\Y$ only, if the points on the boundaries of $\slab{\u}$ are disregarded. One such slab $\slab{\u}$ is displayed with $\u = A_1$; this slab contains $2$ data points in its interior ($B$ and $D$), and $1$ data point outside of it ($C$). Thus, to the minimum in Step~\ref{S4} of our algorithm it contributes by value $1$. The smaller brown points $M_1$ -- $M_{12}$ on $\Sph[1]$ are the midpoints of the sequence of the ordered (anti-)circles, used in the second variant of our algorithm explained in Remark~\ref{remark:d2 algorithm} in Section~\ref{section:bivariate algorithm}.}
    \label{fig:d2}
\end{figure}

\begin{figure}[htpb]
    \centering
    \includegraphics[width=.9\textwidth]{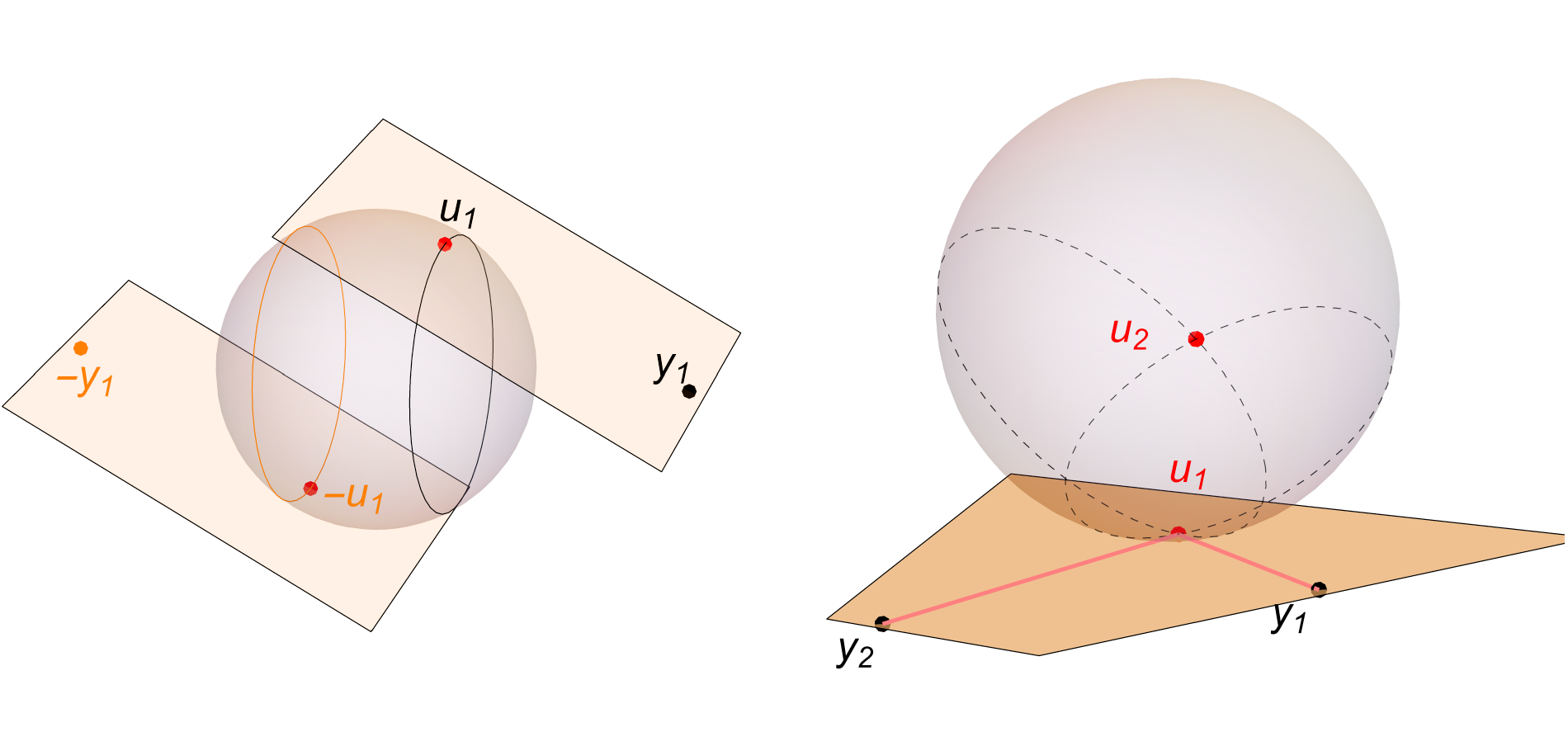}
    \caption{The situation in dimension $d=3$: Left: A single data point $\y_1 = \left(2,0,0\right)\tr$ (black) and its antipodal reflection $-\y_1$ (orange). The circle (black curve) and the anti-circle (orange curve) of $\y_1$ are both $1$-spheres, that is simple curves, on the surface of $\Sph[2]$. The tangent planes of $\Sph[2]$ at the points of the circle pass through $\y_1$, while the tangent planes at the points of anti-circle pass through $-\y_1$. Two such tangent planes are displayed; they correspond to $\u_1$ and $-\u_1$, respectively. Right: Two data points $\y_1$ and $\y_2$ (black points) with their circles (dashed curves) on the unit sphere $\Sph[2]$. The circles intersect in two maximal directions (red points) $\u_1, \u_2 \in \Sph[2]$. These directions correspond to the two maximal tangent (hyper)planes of $\Sph[2]$ that pass through both $\y_1$ and $\y_2$. The tangent plane $H_{\u_1}$ is displayed along with the lines from $\u_1$ to $\y_i$, $i=1,2$, inside $H_{\u_1}$.}
    \label{fig:d3}
\end{figure}

\begin{figure}
    \includegraphics[width=.9\textwidth]{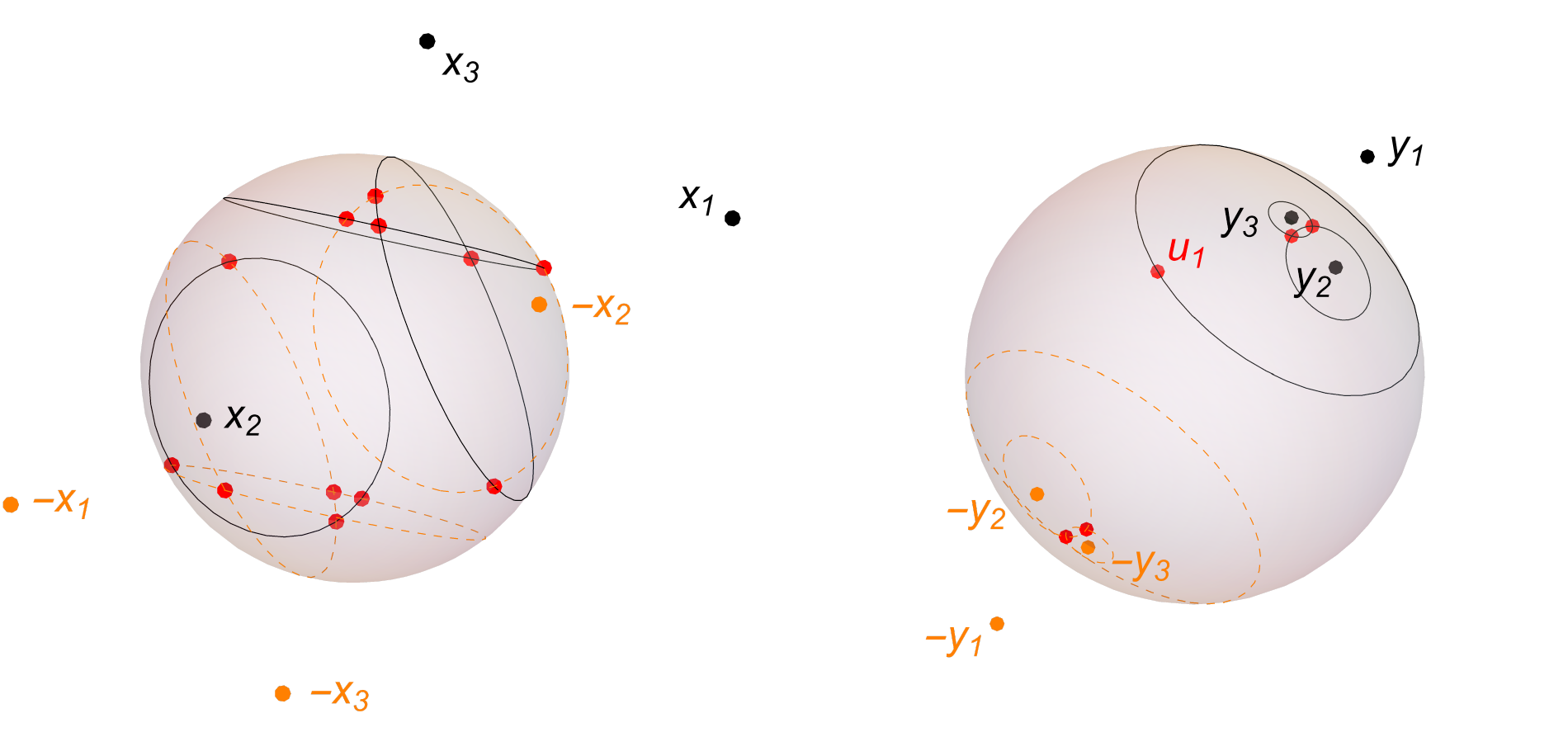}
  \caption{The case $d=3$, with a dataset $\Y$ of three points $\y_1, \y_2, \y_3 \in \R^3$ (black points) and their corresponding circles (black curves) and anti-circles (orange dashed curves) on the unit sphere $\Sph[2]$. Left: Configuration of points where each circle of $\Y$ intersects another (anti-)circle, resulting in two maximal tangent hyperplanes per each intersection of two (anti-)circles. The maximal directions of these (hyper)planes are displayed as the red points on the sphere $\Sph[2]$, and they are exactly the points of $\fdirs$ in Theorem~\ref{theorem:main}. Right: A set of points $\Y$ such that the (anti-)circles of $\y_2$ and $\y_3$ intersect (in the two pairs of red points on $\Sph[2]$, respectively), while none of them intersects the (anti-)circle of $\y_1$ (the largest curves in black and dashed orange, respectively). The set of directions $\fdirs$ has to be chosen as the two points of intersection of circles of $\y_2$ and $\y_3$ (case $k=d-1=2$), and an arbitrary single point $\u_1 \in \Sph[2]$ on the circle of $\y_1$ (red point, case $k=d-2=1$).}
  \label{fig:circles}
\end{figure}

\begin{figure}[htpb]
    \centering
    \includegraphics[width=.9\textwidth]{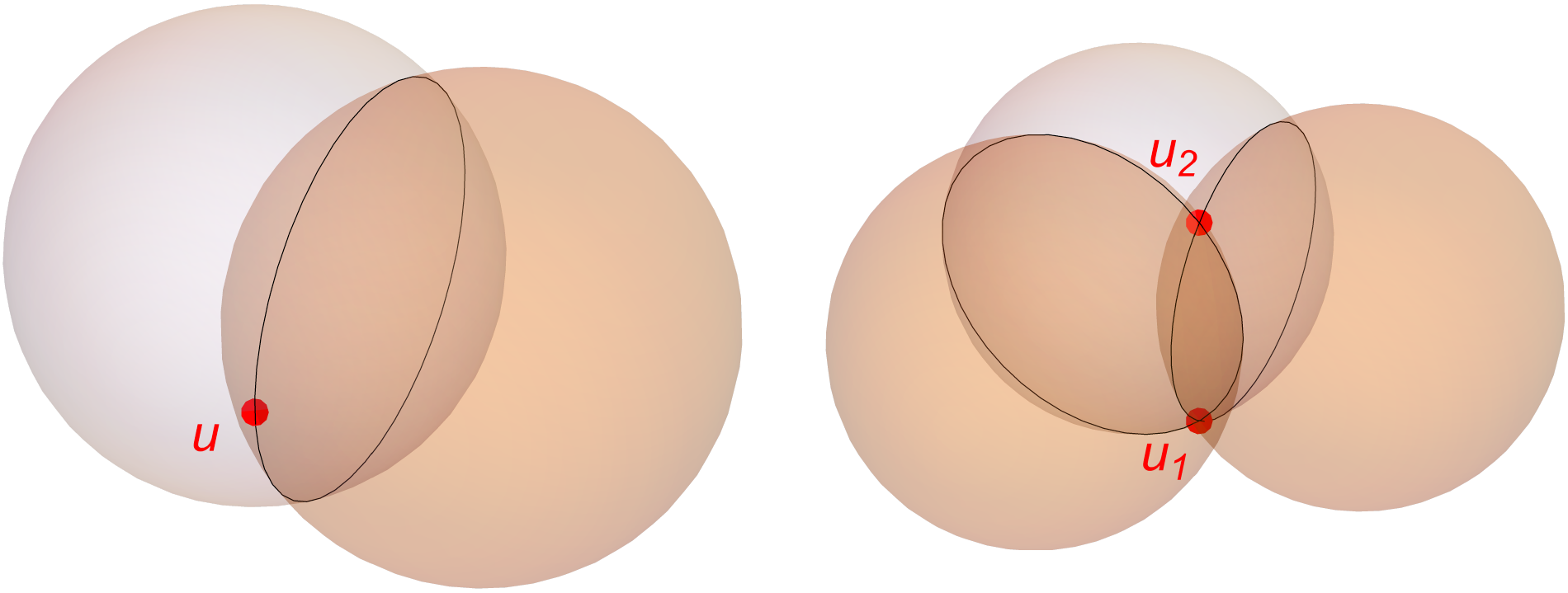}
    \caption{The situation in dimension $d=4$. Left: Two intersecting circles of data points are displayed. Both these circles are $2$-spheres inside the $3$-sphere $\Sph[3]$; here they are displayed canonically embedded into $\R^3$. The intersection of these circles is a $1$-sphere (black curve) that is not intersected by any other (anti-)circle of the dataset. As a maximal direction, an arbitrary direction $\u \in \Sph[3]$ (red point) in the intersection of these two circles is taken in Theorem~\ref{theorem:main} (case $k=d-2=2$). Right: Three circles of points in $\R^4$ that intersect in two maximal directions (a $0$-sphere) displayed as the red points $\u_1, \u_2 \in \Sph[3]$ (case $k=d-1=3$). Again, the setup in $\R^4$ is projected canonically to $\R^3$ in this figure.}
    \label{fig:d4}
\end{figure}

\vskip 0.1 in
\section{Computational algorithms}
\paragraph{}
\vskip 0.1 in \label{Sec:algorithms}

%
%
%
%

Before embarking on the details of the general exact algorithm from Section~\ref{section:idea of algorithm} valid in each dimension $d \geq 2$, we begin by expounding its basic idea in the simpler situation of dimension $d=2$ in Section~\ref{section:bivariate algorithm}. That will help us to fix ideas and develop understanding for the more technical general description in Section~\ref{section:general algorithm}.

\begin{remark}
In all algorithms that follow, $\epsilon$ denotes a small positive constant, that should effectively stand for a ``positive zero'', to account for flopping errors. In our implementations $\epsilon$ is set by default to $10^{-14}$.
\end{remark}

\subsection{Exact algorithm: Bivariate case}  \label{section:bivariate algorithm}

\begin{algo} \label{algorithm:1} \textcolor{\ccol}{\textbf{Inputs:}} $\X = \left\{ \x_1, \dots, \x_n\right\}$, $\bmu \in \R^d$, $\Sigma \in \R^{d \times d}$ positive definite, and $\epsilon > 0$.
  \item[~\stp{1}] Initialize $\textSHD{} = n$ and $\fdirs = \emptyset$.

  \item[~\stp{2}] Transform $\y_i = \Sigma^{-1/2} (\x_i - \bmu)$ for $i = 1, \dots, n$.
  
  \item[~\stp{3}] Find the observations that lie outside $\Sph[1]$ and store them in a list $\Yout = \{\y_{i} \colon \|\y_i\| > 1 \mbox{ and }i \in \left\{ 1, \dots, n \right\}\}$. If $\#\Yout \leq 1$, return $\textSHD{} = 0$ and terminate the algorithm.

  \item[~\stp{4}] For each $\y \in \Yout$, compute its circle composed of the two directions $\u_1, \u_2 \in \Sph[1]$ such that $\u_j\tr(\y-\u_j) = 0$, $j=1,2$. Append the two directions $\u_1$, $\u_2$ of the circle of $\y$, and the two directions of the anti-circle $-\u_1$, $-\u_2$ of $\y$, to the list $\fdirs$.

  \item[~\stp{5}] For each $\vv \in \fdirs$, obtain $p_1 = \#\{i \in \left\{ 1, \dots, n \right\} \colon |\vv\tr \y_i| > 1 + \epsilon\}$ and $p_2 = \#\{i \in \left\{ 1, \dots, n \right\} \colon |\vv\tr \y_i| < 1 - \epsilon\}$. Set $\textSHD{} = \min\{p_1, p_2, \textSHD{}\}$.

  \item[~\textcolor{\ccol}{\textbf{Output:}}] The value of \textSHD{}.
\end{algo}

\begin{remark}
Algorithm~\ref{algorithm:1} is a direct translation of our main procedure from Section~\ref{section:idea of algorithm} in case $d=2$ into pseudocode. In the \proglang{R} package \code{scatterdepth} our code is, however, further optimized for performance and numerical stability. We do the following amendments.
\begin{enumerate}   \label{remark:d2 algorithm}
    \item Once we found in \textbf{Step 3} which observations $\y_i$ are contained inside the unit $1$-sphere, these observations are excluded from the subsequent computation and only the number of them $m \leq n - 1$ is added to $p_2$ in \textbf{Step~5}. This is possible since the observations within $\Sph[1]$ always project inside the slab $\slab{\u}$ for each $\u \in \Sph$.
    \item We use polar coordinates to represent the directions $\fdirs\subset\Sph[1]$. The elements of $\fdirs$ are ordered according to their direction angles $\theta_j \in [-\pi,\pi)$. Instead of the $\epsilon$-perturbation performed in \textbf{Step~5} we take the mid-point $\vv \in \Sph[1]$ on the arc between each two consecutive angles, and use the slab $\slab{\vv}$ directly, see also Figure~\ref{fig:d2}. This way we get $p_1 = \#\{i \in \left\{ 1, \dots, n \right\} \colon |\vv\tr \y_i| > 1 \}$ and $p_2 = \#\{i \in \left\{ 1, \dots, n \right\} \colon |\vv\tr \y_i| < 1 \}$, and we do without the choice of the constant $\epsilon$.
    \item Since each element of $\u \in \fdirs$ is in the list $\fdirs$ together with its opposite direction $-\u$, in the polar domain in $\R^2$ we can restrict to angles $\theta_j \in [0, \pi)$ (Remark~\ref{remark:symmetry}). 
 \end{enumerate}
\end{remark}


%
%

\subsection{Exact algorithm: Multivariate case}   \label{section:general algorithm}

For data of dimension $d > 2$, our algorithm is rather complex. We proceed by presenting its complete pseudocode. A detailed discussion on the individual steps is provided in Comments \textcolor{\ccol}{\textbf{C1}}--\textcolor{\ccol}{\textbf{C8}} after the code. In the code, places where these comments apply are indicated.

\begin{remark}  \label{remark:visited}
In the description of Algorithm~\ref{algorithm:2}, we say that a collection of (anti-)circles corresponding to points from $\Y$ was already \emph{visited} if a maximal direction in their intersection was already found, and added to the list $\fdirs$. According to our discussion in Section~\ref{section:main theorem}, a visited set of \mbox{(anti-)circles} does not correspond to any additional maximal tangent hyperplane that would need to be appended to $\fdirs$, and thus it does not need to be considered. In this sense, once a collection of (anti-)circles has been visited, all its sub-collections will be flagged as visited too. Finally, note that in Algorithm~\ref{algorithm:2} we do not explicitly keep the list of directions $\fdirs$. Instead, the minimum in Step~\ref{S4} of our main algorithm is computed on the fly.
\end{remark}

By $\lin{\u_1, \dots, \u_m} = \left\{ \sum_{i=1}^m \lambda_i \u_i\in \R^d \colon \lambda_i \in \R \mbox{ for all } i \in \left\{ 1, \dots, m \right\} \right\}$
we denote the linear span of the vectors $\u_1, \dots, \u_m \in \R^d$.

\begin{algo} \label{algorithm:2} \textcolor{\ccol}{\textbf{Inputs:}} \label{step0} $\X = \left\{ \x_1, \dots, \x_n\right\}$, $\bmu \in \R^d$, $\Sigma \in \R^{d \times d}$ positive definite, and $\epsilon > 0$.
  \item[~\stp{1}] Initialize $\textSHD{} = n$.  

  \item[~\stp{2}] Transform $\y_i = \Sigma^{-1/2} (\x_i - \bmu)$ for $i = 1, \dots, n$.
  
  \item[~\stp{3}] Compute the list $\Yout = \{\y_{i} \colon \|\y_i\| > 1 \mbox{ and } i \in \left\{ 1, \dots, n \right\} \}$, and denote the number of such data points by $m = \#\Yout$. If $m \leq d -1$, return $\textSHD{} = 0$ and terminate the computation {\upshape{(\textcolor{\ccol}{\textbf{C1}})}}. Without loss of generality, suppose that $\Yout = \left\{ \y_1, \dots, \y_m \right\}$.
  
  \item[~\stpc{4}{Initialization}]  Initialize a series of $d-2$ lists of Boolean values {\upshape{(\textcolor{\ccol}{\textbf{C4}})}}: 
    \begin{itemize}
        \item $\fl{1}$ a vector of length $2\,m$;
        \item $\fl{2}$ a matrix of size $2\,m \times 2\,m$;
        \item $\fl{3}$ a three-dimensional array of size $2\,m \times 2\,m \times 2\,m$;
        \item \dots; 
        \item $\fl{d-2}$ a $(d-2)$-dimensional array of size $2\,m \times 2\,m \times \dots \times 2\,m$.
    \end{itemize}
    All entries of these lists are initialized to \fal. In each dimension of each of these lists, the first $m$ elements correspond to the circles generated by the $m$ points of $\Yout$, and the latter $m$ elements to the anti-circles generated by $\Yout$. The $j$-th list $\fl{j}$ represents indicators of whether $j$-element sets of (anti-)circles were already visited by the algorithm. 
  
    \item[~\stpc{5}{Main loop, case $k = d-1$}] The main loop consists of two nested for-loops. Their setups are described in \textcolor{\ccol}{\textbf{Outer loop}} and~\textcolor{\ccol}{\textbf{Inner loop}}. The batch of work inside the iterated loops is described in~\textcolor{\ccol}{\textbf{Main batch}}.
    
    \item[~\stpf{Outer loop}] For each $(d-1)$-tuple of distinct points $\{\z_1, \dots, \z_{d-1}\} = \left\{ \y_{i_1}, \dots, \y_{i_{d-1}} \right\} \subset \Yout$ with $1 \leq i_1 < \dots < i_{d-1} \leq m$ and $\z_j = \y_{i_j}$, $j=1,\dots,d-1$, run:
  
    \item[~\stpf{Inner loop}] For each vector $\left( s_2, \dots, s_{d-1} \right)\tr$ of $(d-2)$ signs $s_j \in \left\{ -1, +1 \right\}$, $j = 2, \dots, d-1$, replace $\z_j$ by $\z_j^* = s_j \cdot \z_j = s_j \cdot \y_{i_j}$ for $j = 2, \dots, d-1$, while keeping $\z_{1}^* = \z_{1} = \y_{i_1}$ fixed {\upshape{(\textcolor{\ccol}{\textbf{C3}})}}. Denoting by $\Yout^*$ the set of $2\,m$ elements $\y_i$ and $-\y_i$ with $\y_i \in \Yout$ we therefore have $\left\{ \z_1^*, \dots, \z_{d-1}^* \right\} \subset \Yout^*$. Run:
    
    \item[~\stpf{Main batch}] Given the set $\left\{ \z_1^*, \dots, \z_{d-1}^* \right\} \subset \Yout^*$, compute:
            \begin{enumerate}
            \item[(a)] Consider the $(d-2)$-dimensional {\upshape{(\textcolor{\ccol}{\textbf{C5}})}} linear space $\Pi = \lin{\z_2^* - \z_{1}^*, \dots, \z_{d-1}^* - \z_{1}^*}$ parallel to the affine hull of the points $\{ \z_j^* \}_{j=1}^{d-1}$. Obtain the two orthonormal basis vectors $\bm\alpha_1$, $\bm\alpha_2$ of the two-dimensional orthogonal complement $\Pi^\bot$ to $\Pi$.
            \item[(b)] Compute $\w = ((\z_1^*)\tr \bm\alpha_1, (\z_1^*)\tr \bm\alpha_2)\tr \in \R^2$ being the projection of the affine hull of $\{ \z_j^* \}_{j=1}^{d-1}$ into $\Pi^\bot$, and $\w_{i} = (\y_i\tr \bm\alpha_1, \y_i\tr \bm\alpha_2)\tr \in \R^2$, $i = 1, \dots, n$, the projections of the data points into $\Pi^\bot$ {\upshape{(\textcolor{\ccol}{\textbf{C7}})}}. In both cases, the space $\Pi^\bot$ with the basis $\bm\alpha_1$, $\bm\alpha_2$ is canonically identified with the Euclidean space $\R^2$ with the basis $\left(1,0\right)\tr$, $\left(0,1\right)\tr$. 
            \item[(c)] If $\|\w\| < 1$, the collection of points $\left\{ \z_1^*, \dots, \z_{d-1}^* \right\}$ is not a maximal subset of $\Y^*$  {\upshape{(\textcolor{\ccol}{\textbf{C6}})}}. In that case, go to the next iteration of the \textcolor{\ccol}{\textbf{Inner loop}} in \textcolor{\ccol}{\textbf{Step~5}} of the present algorithm. Otherwise, we have $\left\Vert \w \right\Vert > 1$ and we found a new maximal tangent hyperplane.\footnote{Note that $\|\w\| = 1$ cannot happen due to assumption of the general position w.r.t. the Mahalanobis ellipsoid, see Definition~\ref{definition:general position}.} Then run:
                \begin{enumerate}[label=(\roman*)]
                    \item \textcolor{\ccol}{\textbf{Flagging visited (anti-)circles.}} Flag as $\tru$ all elements of the lists $\fl{1}, \dots, \fl{d-2}$ that correspond to the collection of the (anti-)circles associated with $\z_j^*$, $j=1,\dots,d-1$, visited in the current loop of the algorithm (see Remark~\ref{remark:visited}). Precisely, with the convention that $s_1 = +1$, we denote for all $i_j$, $j=1,\dots,d-1$, 
                        \begin{equation}
                            \label{i indices}
                            i_j^* = \begin{cases}
                                    i_j & \mbox{ if }s_j = +1, \\
                                    i_j + m & \mbox{ if }s_j = -1,
                                    \end{cases}
                        \end{equation}  
                    and run the following {\upshape{(\textcolor{\ccol}{\textbf{C4}})}}:
                        \begin{itemize}
                            \item For $d=2$, we have $d-2 = 0$, and no action is necessary.
                            \item For $d=3$, we set $\fl{1}[i_1^*] = \fl{1}[i_2^*] = \tru$ to indicate that the (anti-)circles of $\z_j^* = s_j \cdot \y_{i_j}$, $j=1,2$, were already visited by our algorithm.
                            \item For $d=4$, set $\fl{1}[i_1^*] = \fl{1}[i_2^*] = \fl{1}[i_3^*] = \tru$ as for $d=3$. In addition, for each of the $\binom{3}{2} = 3$ two-point subsets $\left\{ \ell_1^*, \ell_2^* \right\}$ of $\left\{i_1^*, i_2^*, i_3^*\right\}$ fill out $\fl{2}[\ell_1^*, \ell_2^*] = \fl{2}[\ell_2^*, \ell_1^*] = \tru$.
                            \item In general, for $d \geq 2$, for each $k \in \left\{ 1, \dots, d-2\right\}$ and for all $k$-element subsets $\left\{\ell_1^*, \dots,\ell_k^*\right\}$ of $\left\{i_1^*, \dots, i_{d-1}^*\right\}$ set to $\tru$ all the elements of $\fl{k}$ whose indices correspond to a permutation of $\ell_1^*, \dots, \ell_k^*$.
                        \end{itemize}  
                    \item Compute the two $2$-dimensional unit normal vectors $\vv_j$, $j=1,2$ of the tangent line of the $1$-sphere $\Sph[1]$ that passes through $\w$. That is, compute the two directions $\vv_1, \vv_2 \in \Sph[1]$ such that $\vv_j\tr(\w-\vv_j)=0$, $j=1,2$  {\upshape{(\textcolor{\ccol}{\textbf{C7}})}}.
                    \item Compute the quantities $p_{1,j} = \#\{i \in \left\{ 1, \dots, n \right\} \colon |\vv_j\tr \w_{i}| > 1 + \epsilon\}$ and $p_{2,j} = \#\{i \in \left\{ 1, \dots, n \right\} \colon |\vv_j\tr \w_{i}| < 1 - \epsilon\}$ and set $p_j = \min\{ p_{1,j}, p_{2,j}\}$, $j=1,2$.
                    \item Set $\textSHD{} = \min\{p_1, p_2, \textSHD{}\}$.
                \end{enumerate}
        \end{enumerate}
  \item[~\stpc{6}{Main loop, case $k < d-1$}]  For all $k \in \left\{ d-2, \dots, 1 \right\}$ in decreasing order {\upshape{(\textcolor{\ccol}{\textbf{C2}})}} run through the list $\fl{k}$: 
    \begin{enumerate}
        \item[(a)] Consider each entry of $\fl{k}$ that takes value \fal, meaning that the corresponding collection of $k$ (anti-)circles was never visited in \textcolor{\ccol}{\textbf{Step~5}}(c) or in previous iterations of \textcolor{\ccol}{\textbf{Step~6}} of this algorithm (Remark~\ref{remark:visited}). The indices $i_1^*, \dots, i_k^*$ of each such entry correspond to either circles if $i_j^* \leq m$, or to anti-circles if $i_j^* > m$, $j=1,\dots, k$. Define the index $i_j$ and the sign $s_j$ given by
        \[ i_j =   \begin{cases}
                i_j^* & \mbox{ if }i_j^*\leq m, \\
                i_j^*-m & \mbox{ if }i_j^* > m,
                \end{cases}
                \quad
        s_j = \begin{cases}
            +1 & \mbox{ if }i_j^* \leq m, \\
            -1 & \mbox{ if }i_j^* > m.
            \end{cases}
                \]
        The map $i_j^* \mapsto i_j$ is inverse to \eqref{i indices}. Let $\z_{j}^* = s_j \cdot \y_{i_j}$ be the point corresponding to the circle or the anti-circle with index $i_j$. Take the $k$-tuple of points $\left\{\z_1^*, \dots, \z_k^*\right\} \subset \Y^*$.
        \item[(b)] Consider the $(k-1)$-dimensional {\upshape{(\textcolor{\ccol}{\textbf{C5}})}} linear space $\Pi_k = \lin{\z_2^* - \z_{1}^*, \dots, \z_{k}^* - \z_{1}^*}$ parallel to the affine hull of the points $\{ \z_j^* \}_{j=1}^{k}$. Complete $\Pi_k$ by attaching to it the linear span of $d-k-1$ vectors $\vv_1, \dots, \vv_{d-k-1} \in \Sph$ that form (any) orthogonal collection of vectors in the orthogonal complement $\Pi_0^\bot$ of the at most $k$-dimensional linear space $\Pi_0 = \lin{\z_1^*, \dots, \z_k^*}$ {\upshape{(\textcolor{\ccol}{\textbf{C8}})}}. Together, the vectors $\z_2^* - \z_{1}^*, \dots, \z_{k}^* - \z_{1}^*, \vv_1, \dots, \vv_{d-k-1}$ form a $(d-2)$-dimensional linear space 
            \[  \Pi = \lin{\z_2^* - \z_{1}^*, \dots, \z_{k}^* - \z_{1}^*,  \vv_1, \dots, \vv_{d-k-1}}. \]
        Obtain the two orthonormal basis vectors $\bm\alpha_1$, $\bm\alpha_2 \in \Sph$ of the two-dimensional orthogonal complement $\Pi^\bot$ to $\Pi$  {\upshape{(\textcolor{\ccol}{\textbf{C7}})}}. 
        \item[(c)] Compute $\w = ((\z_1^*)\tr \bm\alpha_1, (\z_1^*)\tr \bm\alpha_2)\tr \in \R^2$ being the projection of the affine hull of $\{ \z_j^* \}_{j=1}^{k}$ into $\Pi^\bot$, and $\w_{i} = (\y_i\tr \bm\alpha_1, \y_i\tr \bm\alpha_2)\tr \in \R^2$, $i = 1, \dots, n$, the projections of the data points into $\Pi^\bot$ {\upshape{(\textcolor{\ccol}{\textbf{C7}})}}. In both cases, the space $\Pi^\bot$ with the basis $\bm\alpha_1$, $\bm\alpha_2$ is canonically identified with the Euclidean space $\R^2$ with the basis $\left(1,0\right)\tr$, $\left(0,1\right)\tr$.
        \item[(d)] If $\|\w\| < 1$, the collection of points $\left\{ \z_1^*, \dots, \z_{k}^* \right\}$ is not a maximal subset of $\Y^*$  {\upshape{(\textcolor{\ccol}{\textbf{C6}})}}. In that case, go to the next iteration of \textcolor{\ccol}{\textbf{Step~6}} of the present algorithm. Otherwise, we have $\left\Vert \w \right\Vert \geq 1$ and we found a new maximal tangent hyperplane. Then run:
            \begin{enumerate}[label=(\roman*)]
                \item Flag as $\tru$ all subsets of the current collection of circles and anti-circles visited in the present iteration of \textcolor{\ccol}{\textbf{Step~6}} in the same way as in \textcolor{\ccol}{\textbf{Step~5}}(c)(i).
                \item Compute a $2$-dimensional unit normal vector $\vv\in\Sph[1]$ of (any) tangent line of the $1$-sphere $\Sph[1]$ that passes through $\w$, i.e. find any $\vv \in \Sph[1]$ such that $\vv\tr(\w-\vv)=0$.
                \item Compute $p_{1} = \#\{i \in \left\{ 1, \dots, n \right\} \colon |\vv\tr \w_{i}| > 1 + \epsilon\}$ and $p_{2} = \#\{i \in \left\{ 1, \dots, n \right\} \colon |\vv\tr \w_{i}| < 1 - \epsilon\}$ and set $p = \min\{ p_{1}, p_{2}\}$.
                \item Set $\textSHD{} = \min\{p, \textSHD{}\}$.
            \end{enumerate}        
    \end{enumerate}
  \item[~\textcolor{\ccol}{\textbf{Output:}}] The value of \textSHD{}.
\end{algo}


We see that the complete description of the general Algorithm~\ref{algorithm:2} is rather contrived. Several explanatory comments are in order.

\subsubsection*{\textcolor{\ccol}{\ding{228} C1:} Step~3 and early termination}

In \textbf{Step~3} we halt the computation immediately if the number of observations outside the unit sphere does not exceed $d-1$. Indeed, if there are at most $d-1$ of such data points in $\Y$, together with the origin they lie in a hyperplane in $\R^d$ with unit normal $\u \in \Sph$. The slab $\slab{\u}$ then contains all the data points from $\Y$ in its interior, and $\cslab{\u}$ does not contain any observation, giving $\textSHD{} = 0$.

Further, if at any time during the run the current value of \textSHD{} reaches $0$, the algorithm is terminated and $\textSHD{} = 0$ is returned. For simplicity, this is not indicated in Algorithm~\ref{algorithm:2}.

\subsubsection*{\textcolor{\ccol}{\ding{228} C2:} The order of loops in Steps~5 and~6} 

Recall that we explore maximal tangent hyperplanes that contain $k\in\{1,\dots,d-1\}$ data points. If a set $A\subset\Y^*$ of $k$ points is a maximal subset of $\Y^*$, then no strict subset of $A$ can be maximal. For that reason, we pass through $k$-point subsets with $k$ decreasing from $d-1$ (\textbf{Step~5}) and $d-2$ (\textbf{Step~6}) to $1$ (\textbf{Step~6}), and we only take those combinations of $k$ points that were not previously visited (as part of a bigger maximal subset of $\Y^*$). Note that formally, \textbf{Steps~5} and~\textbf{6} could be merged in the description of Algorithm~\ref{algorithm:2} and explained together. \textbf{Step~5} is only a particular case of \textbf{Step~6} for $k = d-1$. We chose to explain \textbf{Step~5} separately for sake of clarity, as the geometric ideas are easier to understand for $k=d-1$.

\subsubsection*{\textcolor{\ccol}{\ding{228} C3:} Steps~5 and~6: Signs attached to data points}

By keeping the first element $\z_1^* = \z_1$ fixed, in the inner loop of \textbf{Step~5} we guarantee that we do not explore redundant combinations of points $\z_1^*, \dots, \z_{d-1}^*$. This is true because each combination of signs $\left(s_1, \dots, s_{d-1} \right)\tr \in \left\{ -1 , +1 \right\}^{d-1}$ results in the same values $p_1$, $p_2$ in \textbf{Step~5}(c) as the combination of the opposite signs $\left(-s_1, \dots, -s_{d-1} \right)\tr$, thanks to Remark~\ref{remark:symmetry}. Keeping $s_1 = +1$ in all cases, we reduce the complexity of the loop in \textbf{Step~5} of the algorithm to a half. The same restriction is applied in \textbf{Step~6}, although for brevity this is not indicated in the formal description of Algorithm~\ref{algorithm:2}.

\subsubsection*{\textcolor{\ccol}{\ding{228} C4:} Step~4 and economical representation of Boolean lists}

Our representation in the series of $d-2$ lists $\fl{1}, \dots, \fl{d-2}$ is redundant, and was introduced as such only for ease of exposition. 
The implementation of the lists $\fl{1}, \dots, \fl{d-2}$ used in our $\proglang{C++}$ function is much more economical. All the lists are stored inside a single $(d-2)$-dimensional array $\flc$ of dimensions $2\,m \times 2\,m \times \dots \times 2\,m$. Even in this structure, only some elements are effectively used:
    \begin{itemize}
        \item The first $m$ elements of $\fl{1}$, corresponding to the circles of $\y_i$ are the entries $\flc[i,i,\dots,i]$, $i=1,\dots,m$. Due to our discussion in Comment~\textcolor{\ccol}{\textbf{C3}} we do not need to consider the anti-circles of individual points, i.e. indices $i \leq m$ are sufficient for $\fl{1}$.
        \item For the list $\fl{2}$, by~\textcolor{\ccol}{\textbf{C3}} again, we need to store only the information whether the collection of points $\y_{i_1}$ and $\y_{i_2}$ with $1 \leq i_1 < i_2 \leq m$ was visited, with $\y_{i_1}$ corresponding to a circle. The circle of $\y_{i_1}$ and the circle of $\y_{i_2}$ correspond to $\flc[i_1,i_1,\dots,i_1,i_2]$; the circle of $\y_{i_1}$ and the anti-circle of $\y_{i_2}$ are $\flc[i_1,i_1,\dots,i_1,i_2+m]$. This way, only entries $\flc[i_1, \dots, i_1, i_2]$ are used with $i_1 \in \left\{ 1, \dots, m \right\}$ and $i_2 \in \left\{ i_1+1, \dots, m \right\} \cup \left\{i_1 + 1 + m, \dots, 2\,m \right\}$ for encoding $\fl{2}$.
        \item Analogously, for triples and $\fl{3}$, the information about the collection of the circle of $\y_{i_1}$ and the (anti-)circles of $\y_{i_2}$ and $\y_{i_3}$ with $1 \leq i_1 < i_2 < i_3 \leq m$ is stored in the element $\flc[i_1, i_1, \dots, i_1, i_2 + s_2\,m, i_3 + s_3\, m]$, with $s_j = 0$ if the circle of $\y_{i_j}$ is considered and $s_j = 1$ if the anti-circle of $\y_{i_j}$ is taken, $j=2,3$.
        \item $\dots$
        \item Finally, a set of $d-2$ distinct indices $1 \leq i_1 < i_2 < \dots < i_{d-2} \leq m$ is represented as $\flc[i_1, i_2 + s_2\,m, \dots, i_{d-2} + s_{d-2}\, m]$ with $s_j \in \{0,1\}$ for all $j \in \left\{ 2, \dots, d-2 \right\}$ as in the previous cases.
    \end{itemize}    
This allows us to store all the necessary information about the collections of (anti-)circles in a single data structure $\flc$.

\subsubsection*{\textcolor{\ccol}{\ding{228} C5:} Steps~5 and 6: Dimensionality of affine spaces}

In \textbf{Steps~5}(a) and~\textbf{6}(b) we claimed that the dimension of the space $\Pi$ is $d-2$, and the dimension of $\Pi_k$ is $k-1$, respectively. Indeed, our assumption of data points $\X$ being in general position from Definition~\ref{definition:general position} guarantees that any affine hull $A$ of $k$ points from $\Y^*$ that does not intersect the unit sphere $\Sph$ must be of dimension $k-1$. For if it was not, it would be possible to find a hyperplane $H$ tangent to $\Sph$ that contains $A$, which means that inside $H$ there is a set of $k \leq d-1$ points of $\Y^*$ that are not in general position.

\subsubsection*{\textcolor{\ccol}{\ding{228} C6:} Steps~5(c) and 6(d): Intersecting (anti-)circles and affine hulls}

We need to be able to find intersections of (anti-)circles analytically. To compute those points efficiently, we state an equivalence of \begin{enumerate*}[label=(\roman*)] \item the fact that (anti-)circles have non-empty intersection, with \item a certain property of the affine hull of the corresponding sample points.\end{enumerate*} To simplify our argument, consider only circles; the case of anti-circles is dealt with completely analogously. The proof of the following lemma can be found in the Appendix, Section~\ref{section:proof intersection1}. 

\begin{lemma}   \label{lemma:intersection1}
A set of $k \in \left\{ 1, \dots, d-1 \right\}$ circles has a non-empty intersection in $\Sph$ if and only if the affine hull of the corresponding sample points from $\Y$ does not intersect the interior of the unit ball $B = \left\{ \x \in \R^d \colon \left\Vert \x \right\Vert < 1 \right\}$.
\end{lemma} 

This result is not applied directly in $\R^d$, but rather in combination with a projection into $\R^2$ as discussed in Comment~\textcolor{\ccol}{\textbf{C7}}.

\subsubsection*{\textcolor{\ccol}{\ding{228} C7:} Steps~5 and~6: Projections into two-dimensional spaces}

An important nuance of Algorithm~\ref{algorithm:2} concerns the way the maximal directions $\u_j \in \fdirs$ are found in \textbf{Steps~5} and~\textbf{6}. Suppose that $\u_j$ lies in the intersection of $k$ circles ($k=d-1$ for \textbf{Step~5}, or $k<d-1$ for \textbf{Step~6}) corresponding to points $\y_1, \dots, \y_k$ (without loss of generality). We use a projection approach to determine $\u_j \in \Sph$. We first project our data into a two-dimensional subspace $\Pi^\bot$ orthogonal to the affine hull of $\y_1, \dots, \y_k$. There are two situations to be distinguished:
    \begin{enumerate}
        \item In \textbf{Step~5} we have $k = d-1$. The affine hull of $\y_1, \dots, \y_k$ is, according to Comment~\textcolor{\ccol}{\textbf{C5}}, an affine space of dimension $d-2$, and $\Pi^\bot$ is simply its orthogonal complement.
        \item If \textbf{Step~6} we have $k<d-1$. Then, the points $\y_1, \dots, \y_k$ are complemented by a set $V$ of (arbitrarily chosen) vectors such that the complete set $V \cup \left\{ \y_1, \dots, \y_k \right\}$ generates an affine space of dimension $d-2$. The way the complementary set $V$ is found is explained in detail in Comment~\textcolor{\ccol}{\textbf{C8}}. The space $\Pi^\bot$ is again taken as the orthogonal complement to this affine space. 
    \end{enumerate}  
    
Both orthogonal complements in \textbf{Steps~5} and~\textbf{6}, and the completion of a subspace by orthogonal vectors in \textbf{Step~6}, are obtained using the QR-decomposition of the matrix whose $k-1$ columns are the vectors $\y_1 - \y_k, \dots, \y_{k-1} - \y_k$ \citep[Section~2.1]{Horn_Johnson2013}.

After projecting our setup into a two-dimensional space $\Pi^\bot$ (identified with $\R^2$), the whole affine hull of $\y_1, \dots, \y_k$ projects to a single point $\w$ in $\Pi^\bot$, which lies outside the open unit ball in $\R^2$ if and only if the intersection of the circles of $\y_1, \dots, \y_k$ is non-empty, see Comment~\textcolor{\ccol}{\textbf{C6}}. For this point $\w$, the two directions $\vv_1, \vv_2 \in \Sph[1]$ inside the two-dimensional space $\Pi^\bot$ such that the tangent lines of $\Sph[1]$ at $\vv_1, \vv_2$ pass through $\w$ are found by direct calculation as vectors $\vv_1, \vv_2 \in \Sph[1]$ such that $\vv_j\tr (\w - \vv_j) = 0$. The final vector $\u_j \in \Sph$ is the unique direction in $\R^d$ that is projected onto $\vv_j \in \Pi^\bot$. In our implementation, however, the final vectors $\u_j$, $j=1,2$, are not needed to be found, since we only need the number of data points $\y_i$ strictly inside (or outside) the slab $\slab{\u_j}$ in~\eqref{equation:main theorem}. This can be inferred directly inside the two-dimensional space $\Pi^\bot$; it is the number of the projected points $\w_i$ strictly inside (or outside, respectively) the two-dimensional slab given by $\vv_j$. For an illustration in dimension $d=3$ see Figure~\ref{fig:C7}.

Our projection technique into a two-dimensional space $\Pi^\bot$ always results in exactly two maximal directions $\u_1, \u_2 \in \Sph$ (or, more precisely, their equivalent representatives $\vv_1, \vv_2 \in \Sph[1]$). They are both needed if $k=d-1$, because in that case we want to find both maximal directions of the $0$-sphere of intersection of the (anti-)circles, see Theorem~\ref{theorem:main} and Lemma~\ref{lemma:intersection} in the Appendix. This is done in \textbf{Step~5}(c). If $k<d-1$, it is enough to take any single such maximal direction as we do in \textbf{Step~6}(d).

\begin{figure}[htpb]
    \centering
    \includegraphics[width=.9\textwidth]{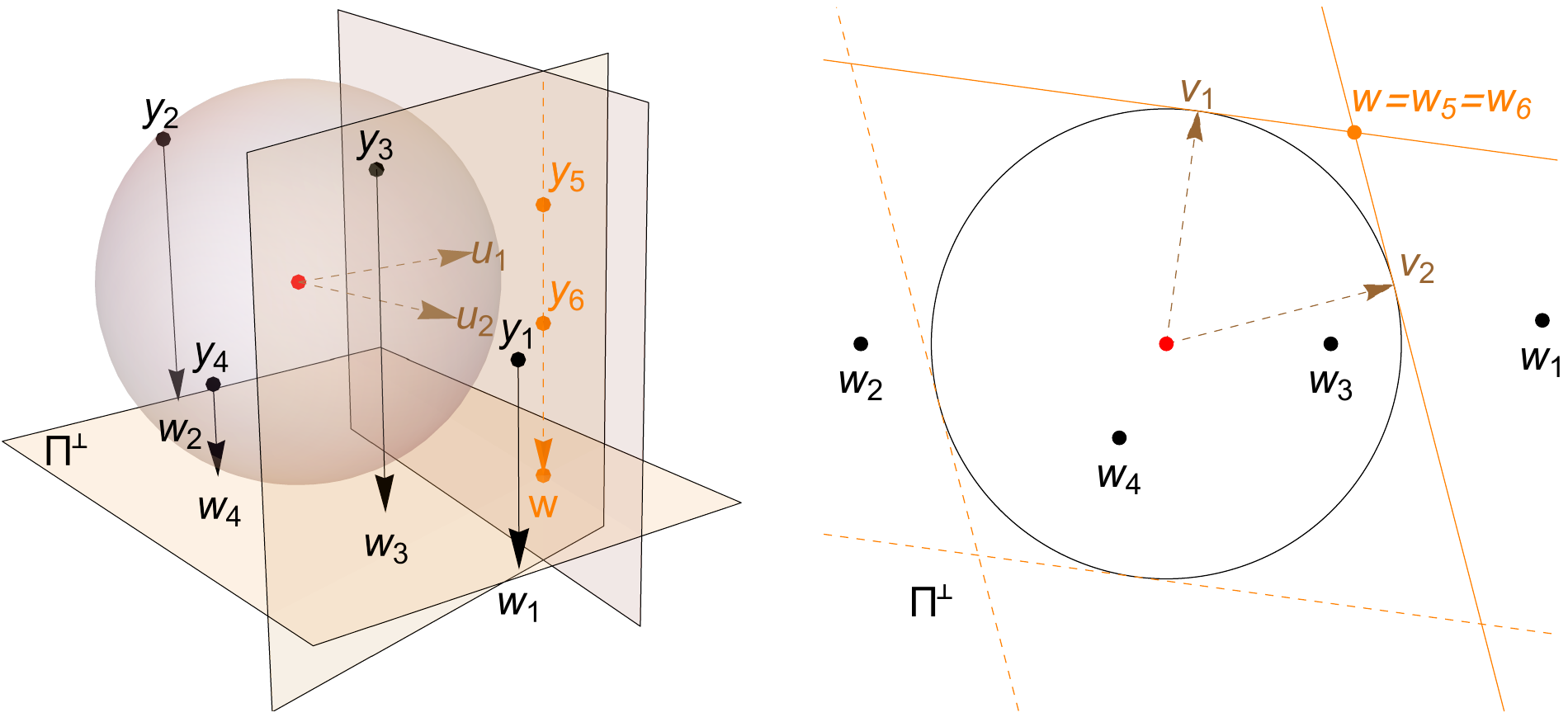}
    \caption{The principle of projection from Comment~\textcolor{\ccol}{\textbf{C7}} in dimension $d=3$. Left: A dataset $\Y$ of six observations $\y_1, \dots, \y_6$. In \textbf{Step~5} we now consider $i_1 = 5$, $i_2 = 6$, and take $\z_1^* = \y_5$, $\z_2^* = \y_6$. The linear space $\Pi$ is parallel with the line spanned by $\y_5$ and $\y_6$ (dashed orange line). The orthogonal space $\Pi^\bot$ is the horizontal plane, where all data points project (solid arrows); the situation inside the plane $\Pi^\bot$ is visualized on the right hand panel of the figure. The projection $\w$ of both $\z_1^*$ and $\z_2^*$ into $\Pi^\bot$ lies outside the unit sphere $\Sph[1]$, meaning that by Lemma~\ref{lemma:intersection1}, two maximal directions $\u_1, \u_2 \in \Sph$ are found (brown arrows) in the left hand figure. According to our discussion in Comment~\textcolor{\ccol}{\textbf{C7}}, these correspond to the only directions that project into $\vv_1, \vv_2 \in \Sph[1]$ in $\Pi^\bot$. In Algorithm~\ref{algorithm:2}, directions $\u_1, \u_2 \in \Sph$ are not computed, as $\y_i \in \slab{\u_j}$ if and only if $\w_i \in \slab{\vv_j}$ for each $i = 1,\dots,n$ and $j = 1,2$. Thus, once projected into $\Pi^\bot$, only $\vv_1$ and $\vv_2$ need to be found.}
    \label{fig:C7}
\end{figure}

\subsubsection*{\textcolor{\ccol}{\ding{228} C8:} Step~6: Expanding the affine hull} 


We are given a $(k-1)$-dimensional\footnote{This follows from Definition~\ref{definition:general position}, thanks to the derivation in Comment~\textcolor{\ccol}{\textbf{C5}}.} affine hull $A$ of the points $\{ \z_j^* \}_{j=1}^k$. We know that $A$ intersects the set $B = \left\{ \x \in \R^d \colon \left\Vert \x \right\Vert < 1 \right\}$ if and only if the (anti-)circles corresponding to $\z_1^*, \dots, \z_k^*$ intersect, due to Comment~\textcolor{\ccol}{\textbf{C6}}. To be able to project into a two-dimensional orthogonal complement in \textbf{Step~6}(c) of the algorithm, we first add to the affine space $A$ the linear space $C$ given as the span of $(d-k-1)$ orthonormal vectors $\vv_1, \dots, \vv_{d-k-1}$ from the orthogonal complement $\Pi_0^\bot$ to the linear space $\Pi_0 = \lin{\z_1^*, \dots, \z_k^*}$. This forms a $(d-2)$-dimensional affine space 
    \begin{equation}\label{equation:A'}
    A' = A + C = \left\{ \a + \c \colon \a \in A \mbox{ and }\c \in C \right\}.
    \end{equation}
For the proof of the next lemma see the Appendix, Section~\ref{section:proof lemma2}.
    
\begin{lemma}   \label{lemma:2}
For the affine space $A'$ from~\eqref{equation:A'} we have that $A \subset A'$. Further, $A'$ intersects $B = \left\{ \x \in \R^d \colon \left\Vert \x \right\Vert < 1 \right\}$ if and only if $A$ intersects $B$. In particular, by Comment~\textcolor{\ccol}{\textbf{C6}}, the circles of the points $\{ \z_j^* \}_{j=1}^k$ intersect in $\Sph$ if and only if $A'$ intersects $B$.
\end{lemma}

Finally, by considering an arbitrary set of vectors $\vv_1, \dots, \vv_{d-k-1}$ in \textbf{Steps~6}(c) and (d) we find an arbitrary point (see Theorem~\ref{theorem:main}) in the intersection of the associated (anti-)circles, if such a point exists. It exists if and only if in \textbf{Step~6}(d) we have $\left\Vert \w \right\Vert < 1$ by Comment~\textcolor{\ccol}{\textbf{C6}}.

%
%


%
%
%

\subsection{Approximate algorithms} \label{section:approximate algorithms}

We will see in Section~\ref{Sec:Illustrations} that especially in higher dimensions $d > 4$ and sample sizes starting at hundreds of points, the exact computation of \textSHD{} depth becomes slow. For that reason, we propose two approximate computational algorithms: 
    \begin{itemize}
    \item An algorithm based on randomly sampled directions $\u_j \in \Sph$, $j = 1, \dots, N$, distributed uniformly in the unit sphere $\Sph$. The approximation of $\SHD(\Sigma, \bmu, \X)$ is
        \[ \min_{j=1,\dots,N} \min\left\{ \# \left\{ i \colon |\u_j\tr \Sigma^{-1/2}(\x_i - \bmu)| \le 1 \right\}, \# \left\{ i \colon |\u_j\tr \Sigma^{-1/2}(\x_i - \bmu)| \ge 1 \right\} \right\}. \]
    In this formula we replace the whole sphere $\Sph$ in~\eqref{equation:scatter depth} by the finite set of randomly chosen points $\left\{ \u_1, \dots, \u_N \right\}\subset \Sph$ and use affine invariance of \textSHD{} as discussed in Section~\ref{section:affine invariance}.
    \item A random extension of Algorithm~\ref{algorithm:2}, where instead of running through all $\binom{n}{d-1}$ tuples of sample points in the outer loop of \textbf{Step~5}, only a random subset of $N$ of the $(d-1)$-tuples of distinct points is selected in \textbf{Step~5}. The search for maximal directions in less than $d-1$ (anti-)circles in \textbf{Step~6} is not executed in this approximation.
    \end{itemize}
In the \proglang{R} package \pkg{scatterdepth} the first approximation method is called \code{rdirections} (standing for \emph{random directions}) and the second \code{rpoints} (standing for \emph{random points}). The method \code{rdirections} can be seen as an extension of the idea of approximate computation applicable to any depth satisfying the projection property \cite[Section~6]{Dyckerhoff2004}. Especially in higher dimensions, the method \code{rdirections} is not expected to perform well; that was observed already for the location halfspace depth in \cite{Nagy_etal2020} and \cite{Dyckerhoff_etal2021}. Both approximation methods are compared in a numerical study in Section~\ref{Sec:Illustrations}.

\subsection{Implementation}

We have implemented both exact Algorithm~\ref{algorithm:1} for $d = 2$ and Algorithm~\ref{algorithm:2} with $d \geq 1$ in the \proglang{R} package \pkg{scatterdepth}, along with both approximate algorithms from Section~\ref{section:approximate algorithms}. The \proglang{R} package is available as part of the Supplementary Materials accompanying this paper. Its latest version is available at \url{https://github.com/NagyStanislav/scatterdepth}.

\vskip 0.1 in
\section{Numerical studies}
\paragraph{}
\vskip 0.1 in \label{Sec:Illustrations}

We conclude by presenting two numerical experiments to assess the performance of our \proglang{R} programs for the computation of \textSHD{}. We study the computational speed of the exact algorithm in Section~\ref{section:speed} as well as the quality and speed of the two approximate algorithms in Section~\ref{section:approximate speed}. The computations were performed on a modest laptop with processor Intel(R) Core(TM) i7-7500U CPU @2.70GHz (2 CPUs), $\sim$ 2.9GHz, 16.0 GB RAM.

\subsection{Speed of the exact algorithm} \label{section:speed}
\paragraph{}

We use simulated datasets of sizes $2^k$ with $k \in \left\{ 2, \dots, 14 \right\}$, amounting to sample sizes ranging from $n = 2^5 = 32$ to $n = 2^{14} = 16~384$, and dimensions $d \in \left\{ 2, \dots, 5 \right\}$. Each simulated dataset comes from the standard $d$-variate normal distribution. For each dataset, the scatter halfspace depth of the identity matrix $\I \in \R^{d \times d}$ centered at the origin $\bm 0_d \in \R^d$ is computed. The means and the standard deviations of the execution times taken over $20$ independent runs of our simulation setup can be found in Table~\ref{tab:times}. All runs whose mean execution times exceeded $5~000$ seconds (approximately $80$ minutes) were terminated, and the computation for higher sample sizes $n$ was not performed. As can be seen from Table~\ref{tab:times}, the exact algorithm is very fast in dimension $d=2$, where even the computation of \textSHD{} for more than $10~000$ points does not take longer than one second. With $d=3$, the computation becomes slower with $n > 1~000$, and takes already about $50$~minutes for $10~000$ points. With $d>3$, exact computation is feasible only for lower hundreds of data points.

\begin{table}[ht]
\centering
\resizebox{\textwidth}{!}{%
\begin{tabular}{c|cccc}
           & \multicolumn{1}{c}{                                 $d = 2$} & \multicolumn{1}{c}{                                 $d = 3$} & \multicolumn{1}{c}{                                 $d = 4$} & \multicolumn{1}{c}{                                 $d = 5$} \\ 
   \hline
$n = 32$    & 0.00070  \scriptsize{(0.00035)}      & 0.00121  \scriptsize{(0.00013)}      & 0.01574  \scriptsize{(0.00291)}      & 0.33318  \scriptsize{(0.05634)}      \\ 
  $n = 64$    & 0.00065  \scriptsize{(0.00009)}      & 0.00411  \scriptsize{(0.00056)}      & 0.16892  \scriptsize{(0.02830)}      & 6.89089  \scriptsize{(0.53821)}      \\ 
  $n = 128$   & 0.00126  \scriptsize{(0.00168)}      & 0.02094  \scriptsize{(0.00281)}      & 2.14470  \scriptsize{(0.24781)}      & 158.05592  \scriptsize{(12.29712)}   \\ 
  $n = 256$   & 0.00135  \scriptsize{(0.00018)}      & 0.13000  \scriptsize{(0.00756)}      & 24.24822  \scriptsize{(2.66352)}     & 3~781.61763  \scriptsize{(263.738)} \\ 
  $n = 512$   & 0.00261  \scriptsize{(0.00020)}      & 0.90816  \scriptsize{(0.06754)}      & 338.44751  \scriptsize{(20.2452)}   & ---                \\ 
  $n = 1~024$  & 0.00742  \scriptsize{(0.00041)}      & 6.87280  \scriptsize{(0.47802)}      & 4~924.46641  \scriptsize{(294.792)} & ---                \\ 
  $n = 2~048$  & 0.02533  \scriptsize{(0.00111)}      & 51.41536  \scriptsize{(2.45273)}     & ---                & ---                \\ 
  $n = 4~096$  & 0.09369  \scriptsize{(0.00510)}      & 443.62476  \scriptsize{(36.4937)}   & ---                & ---                \\ 
  $n = 8~192$  & 0.36278  \scriptsize{(0.01103)}      & 3~110.07089  \scriptsize{(174.922)} & ---                & ---                \\ 
  $n = 16~384$ & 1.39010  \scriptsize{(0.04574)}      & ---                & ---                & ---                \\ 
\end{tabular}
}
\caption{Means and standard deviations (in brackets) of execution times of the implementation of the exact algorithm in \proglang{R} package \code{scatterdepth} in seconds. The results depend on the size $n$ of the dataset and its dimension $d$. They are computed from $20$ independent runs; executions in setups hat took on average longer than $5~000$ seconds (approximately $80$~minutes) were halted.} 
\label{tab:times}
\end{table}

\subsection{Accuracy of approximate algorithms} \label{section:approximate speed}
\paragraph{}
 
Since exact computation of \textSHD{} becomes unfeasible for $d>3$ and thousands of data points, we run a second simulation study where also the two approximate algorithms from Section~\ref{section:approximate algorithms} are compared. The simulation setup is similar to that from Section~\ref{section:speed}. We generated $50$ independent replicates of $d$-dimensional standard normal samples of size $n$, with $d \in \left\{2, \dots, 5\right\}$ and $n = 2^k$, $k \in \left\{ 6, \dots, 10 \right\}$. In each setup, the exact depth \textSHD{} and both its approximations from Section~\ref{section:approximate algorithms} were computed for matrix $\I \in \R^{d \times d}$ centered at $\bm 0_d \in \R^d$. The respective computation times and resulting depth values were recorded. Both approximate algorithms are taken with parameter $N = 10^4$. 

\begin{table}[ht]
\centering
\resizebox{\textwidth}{!}{%
\begin{tabular}{cl|cccc}
            &                     & \multicolumn{1}{c}{                               $d = 2$} & \multicolumn{1}{c}{                               $d = 3$} & \multicolumn{1}{c}{                               $d = 4$} & \multicolumn{1}{c}{                               $d = 5$} \\ 
   \hline
$n = 64$   & \code{exact}       & 0.00062  \scriptsize{(0.00009)}      & 0.00406  \scriptsize{(0.00060)}      & 0.17148  \scriptsize{(0.02266)}      & 6.97294  \scriptsize{(0.63412)}      \\ 
             & \code{rdirections} & 0.00471  \scriptsize{(0.00618)}      & 0.00476  \scriptsize{(0.00073)}      & 0.00565  \scriptsize{(0.00065)}      & 0.00653  \scriptsize{(0.00110)}      \\ 
             & \code{rpoints}     & 0.05899  \scriptsize{(0.00735)}      & 0.10046  \scriptsize{(0.02268)}      & 0.14349  \scriptsize{(0.03644)}      & 0.21733  \scriptsize{(0.02032)}      \\ 
  \hdashline
  $n = 128$  & \code{exact}       & 0.00078  \scriptsize{(0.00004)}      & 0.02158  \scriptsize{(0.00354)}      & 1.99720  \scriptsize{(0.24853)}      & 163.786  \scriptsize{(10.4920)}   \\ 
             & \code{rdirections} & 0.00607  \scriptsize{(0.00214)}      & 0.00918  \scriptsize{(0.00544)}      & 0.01026  \scriptsize{(0.00144)}      & 0.01123  \scriptsize{(0.00102)}      \\ 
             & \code{rpoints}     & 0.05675  \scriptsize{(0.00788)}      & 0.12483  \scriptsize{(0.02884)}      & 0.16974  \scriptsize{(0.02397)}      & 0.25411  \scriptsize{(0.01184)}      \\ 
  \hdashline
  $n = 256$  & \code{exact}       & 0.00121  \scriptsize{(0.00007)}      & 0.12428  \scriptsize{(0.01167)}      & 25.2035  \scriptsize{(1.72585)}     & 3798.80  \scriptsize{(278.392)} \\ 
             & \code{rdirections} & 0.01213  \scriptsize{(0.01015)}      & 0.01491  \scriptsize{(0.00269)}      & 0.01906  \scriptsize{(0.00429)}      & 0.02052  \scriptsize{(0.00204)}      \\ 
             & \code{rpoints}     & 0.05953  \scriptsize{(0.00821)}      & 0.15061  \scriptsize{(0.02625)}      & 0.21078  \scriptsize{(0.01922)}      & 0.34689  \scriptsize{(0.03684)}      \\ 
  \hdashline
  $n = 512$  & \code{exact}       & 0.00253  \scriptsize{(0.00018)}      & 0.91853  \scriptsize{(0.07146)}      & 337.755  \scriptsize{(24.0302)}   & ---                \\ 
             & \code{rdirections} & 0.02183  \scriptsize{(0.00622)}      & 0.02828  \scriptsize{(0.00601)}      & 0.03527  \scriptsize{(0.00350)}      & 0.04050  \scriptsize{(0.00450)}      \\ 
             & \code{rpoints}     & 0.06942  \scriptsize{(0.02467)}      & 0.19226  \scriptsize{(0.01792)}      & 0.30083  \scriptsize{(0.04089)}      & 0.51847  \scriptsize{(0.03597)}      \\ 
  \hdashline
  $n = 1~024$ & \code{exact}       & 0.00685  \scriptsize{(0.00037)}      & 6.61904  \scriptsize{(0.44534)}      & 5049.42  \scriptsize{(323.863)} & ---                \\ 
             & \code{rdirections} & 0.04316  \scriptsize{(0.00565)}      & 0.05681  \scriptsize{(0.00337)}      & 0.07099  \scriptsize{(0.00945)}      & 0.08304  \scriptsize{(0.01335)}      \\ 
             & \code{rpoints}     & 0.08122  \scriptsize{(0.01978)}      & 0.29275  \scriptsize{(0.02596)}      & 0.48117  \scriptsize{(0.03766)}      & 0.89339  \scriptsize{(0.09794)}      \\ 
\end{tabular}
}
\caption{Means and standard deviations (in brackets) of execution times 
                      of the implementation of the three versions of our algorithm in 
                      \proglang{R} package \code{scatterdepth} in seconds 
                      (the \code{exact} Algorithm~\ref{algorithm:2} and the two 
                      approximate algorithms \code{rdirections} and \code{rpoints}
                      from Section~\ref{section:approximate algorithms}). 
                      The results depend on the size $n$ of the dataset and its 
                      dimension $d$. They are computed from $50$ independent runs; 
                      executions of the \code{exact} algorithm that took on average longer 
                      than $5~000$ seconds (approximately $80$~minutes) were halted. 
                      The corresponding assessment of precision of the approximation 
                      is in Table~\ref{tab:exactness}.} 
\label{tab:exactness_times}
\end{table}

In Table~\ref{tab:exactness_times} we see the means and the standard deviations of the computation times for all three versions of our procedure: \begin{enumerate*}[label=(\roman*)] \item the exact Algorithm~\ref{algorithm:2} (\code{exact}), \item the approximate algorithm from Section~\ref{section:approximate algorithms} based on uniformly sampled random directions (\code{rdirections}), and \item the approximate algorithm from Section~\ref{section:approximate algorithms} based on a subsampling idea (\code{rpoints}). \end{enumerate*}

\begin{table}[ht]
\centering
\begin{tabular}{cl|cccc}
            &                     & \multicolumn{1}{c}{                       $d = 2$} & \multicolumn{1}{c}{                       $d = 3$} & \multicolumn{1}{c}{                       $d = 4$} & \multicolumn{1}{c}{                       $d = 5$} \\ 
   \hline
$n = 64$   & \code{rdirections} & 0.00105  \scriptsize{(0.98)} & 0.03701  \scriptsize{(0.60)} & 0.12140  \scriptsize{(0.24)} & 0.27977  \scriptsize{(0.04)} \\ 
             & \code{rpoints}     & 0.00000  \scriptsize{(1.00)} & 0.00000  \scriptsize{(1.00)} & 0.00200  \scriptsize{(0.98)} & 0.07804  \scriptsize{(0.50)} \\ 
  \hdashline
  $n = 128$  & \code{rdirections} & 0.00000  \scriptsize{(1.00)} & 0.01884  \scriptsize{(0.50)} & 0.08271  \scriptsize{(0.02)} & 0.18487  \scriptsize{(0.00)} \\ 
             & \code{rpoints}     & 0.00000  \scriptsize{(1.00)} & 0.00000  \scriptsize{(1.00)} & 0.02166  \scriptsize{(0.52)} & 0.08690  \scriptsize{(0.04)} \\ 
  \hdashline
  $n = 256$  & \code{rdirections} & 0.00057  \scriptsize{(0.96)} & 0.01357  \scriptsize{(0.32)} & 0.05313  \scriptsize{(0.00)} & 0.23012  \scriptsize{(0.00)} \\ 
             & \code{rpoints}     & 0.00000  \scriptsize{(1.00)} & 0.00062  \scriptsize{(0.96)} & 0.02592  \scriptsize{(0.10)} & 0.18951  \scriptsize{(0.00)} \\ 
  \hdashline
  $n = 512$  & \code{rdirections} & 0.00054  \scriptsize{(0.92)} & 0.01230  \scriptsize{(0.12)} & 0.08038  \scriptsize{(0.00)} & \multirow{2}{*}{\textcolor{\ccol}{\textbf{2.12} \scriptsize{(0.72/0.18/0.10)}}}        \\ 
             & \code{rpoints}     & 0.00000  \scriptsize{(1.00)} & 0.00439  \scriptsize{(0.50)} & 0.06389  \scriptsize{(0.00)} &       \\ 
  \hdashline
  $n = 1~024$ & \code{rdirections} & 0.00020  \scriptsize{(0.94)} & 0.00942  \scriptsize{(0.02)} & 0.38901  \scriptsize{(0.00)} & \multirow{2}{*}{\textcolor{\ccol}{\textbf{3.44} \scriptsize{(0.78/0.08/0.14)}}}        \\ 
             & \code{rpoints}     & 0.00000  \scriptsize{(1.00)} & 0.00493  \scriptsize{(0.20)} & 0.38061  \scriptsize{(0.00)} &         \\ 
\end{tabular}
\caption{Mean relative differences in computed depths 
        compared to the exact \textSHD{} values, and mean
        percentage of exact evaluations (in brackets) for the two approximate algorithms (\code{rdirections} and \code{rpoints} from Section~\ref{section:approximate algorithms}) implemented in \proglang{R} package \code{scatterdepth}. In the last two cases $d=5$ and $n = 512, 1024$, the \code{exact} \textSHD{} was not computed. Instead, the mean difference of the approximate depths \code{rdirections}$-$\code{rpoints} is given (the number in \textbf{\textcolor{\ccol}{thick colored font}}) and the proportion of runs where the approximate depth \code{rdirections} was higher/equal/lower than \code{rpoints} (in brackets). The corresponding assessment of computation times is in Table~\ref{tab:exactness_times}.} 
\label{tab:exactness}
\end{table}

To assess the quality of the approximation, we consider the relative difference of the computed depths. The relative difference of the exact value \textSHD{} (computed using Algorithm~\ref{algorithm:2}) and its approximation \textSHDA{} (computed using one of the algorithms from Section~\ref{section:approximate algorithms}) is evaluated as
    \[  \frac{\mbox{\textSHDA{}} - \textSHD{}}{\textSHD{}}.   \]
This quantity is always non-negative. The smaller relative difference, the better the approximation algorithm. Mean relative differences over the $50$ performed runs in our simulation study are reported in Table~\ref{tab:exactness}. With them, we provide the percentage of independent runs in which \textSHDA{} $=$ \textSHD{}, that is the proportion of runs where an exact result was obtained. In the case $d=5$ with $n \in \left\{ 512, 1024 \right\}$, when the \code{exact} \textSHD{} was not computed, we present the mean difference between the obtained approximate depths \code{rdirections} and \code{rpoints}, together with the proportion of runs when \code{rdirections} gave a higher/equal/lower result than \code{rpoints}. Positive values of the mean difference of approximate depths favor the method \code{rpoints}. 

Summarizing the results in Tables~\ref{tab:exactness_times} and~\ref{tab:exactness}, the algorithm \code{rdirections} is only slightly faster, but significantly less precise than \code{rpoints}. Interestingly, also in higher dimension the computational cost of \code{rpoints} is negligible; even for $1~000$ data points in dimension $d=5$, \code{rpoints} terminates in under one second. We conclude that the approximate algorithm \code{rpoints} based on our exact procedure is to be favored when the exact computation of \textSHD{} is not attainable. All the results of our simulation studies agree with our analysis from Section~\ref{section:approximate algorithms}, and also concur with the findings of \cite{Nagy_etal2020} for the location halfspace depth.

\appendix

\section{Proof of Theorem~\ref{theorem:main}}   \label{section:proof main}

The exact computation of \textSHD{} amounts to minimizing an objective function $h \colon \Sph \to \left\{0,1, \dots, n \right\}$ from~\eqref{equation:objective function} over $\Sph$. In Section~\ref{section:tangent hyperplanes} we give the first simplification: it is enough to minimize $h$ over the set of neighborhoods of spherical boundaries\footnote{For a subset $A \subset \Sph$ of the unit sphere, the \emph{spherical boundary} of $A$ is defined as
    \[ \bd{\left\{ \lambda\,\x \in \R^d \colon \lambda \in (0,\infty) \mbox{ and }\x \in A \right\}} \cap \Sph. \]
Spherical interior and spherical neighborhood are defined analogously.} $\dirs \subset \Sph$ of certain distinctive spherical polytopes called shells. This is true because the shells partition the unit sphere into regions where the value of $h$ is constant. The shells are bounded by circles and anti-circles (see~\eqref{equation:circle}), subsets of the unit sphere where the tangent hyperplane to $\Sph$ passes through a given data point or its antipodal reflection, respectively. Shells and (anti-)circles are considered in detail in Section~\ref{section:shells}. In Section~\ref{section:tilting} we use shells to reach a second substantial simplification of our minimization task. Instead of minimizing $h$ over neighborhoods of $\dirs$, we modify $h$ slightly, and minimize that new function only over a finite number of elements from $\dirs$. Before identifying such elements, in Section~\ref{section:intersections} we provide an auxiliary lemma on the dimensionality of intersections of (anti-)circles, and finally use all the previous results in Section~\ref{section:maximal tangent hyperplanes} to prove Theorem~\ref{theorem:main}.

\subsection{Searching through tangent hyperplanes} \label{section:tangent hyperplanes}

As the dataset $\Y$ is finite, the objective function $h(\u)$ takes only finitely many different values among $\left\{ 0, 1, \dots, n \right\}$. Moreover, both mappings making up function $h$ in~\eqref{equation:objective function}
    \[
    \u \mapsto \#\left\{ i \in \{ 1, \dots, n\} \colon \left\vert \u\tr \y_i \right\vert \leq 1 \right\}, \quad \mbox{ and }\quad \u \mapsto \#\left\{ i \in \{ 1, \dots, n\} \colon \left\vert \u\tr \y_i \right\vert \geq 1 \right\},
    \]
are locally constant at all $\u \in \Sph$, except for those directions $\u$ for which there is at least one index $i\in \{1,\dots, n\}$ such that $|\u\tr \y_i|=1$. Geometrically, if there is no data point $\y_i$ lying on one of the two boundary hyperplanes $H_{\u}$, $H_{-\u}$ of the slab $\slab{\u}$, there exists a spherical neighborhood 
    \begin{equation}\label{equation:spherical ball} 
    B(\u;\epsilon) = \left\{ \vv \in \Sph \colon  \lVert \u-\vv \rVert < \epsilon \right\}
    \end{equation}
of $\u$ in $\Sph$ for $\epsilon > 0$ small such that for all $\vv \in B(\u;\epsilon)$, the value of $h(\vv)$ is constant.

The collection of all directions $\u \in \Sph$ such that there exists a point $\y_i$ in the boundary of $\slab{\u}$ will be crucial. We denote the set of all such directions by $\dirs$. Note that $\u\in\dirs$ if and only if the tangent hyperplane $H_{\u}$ from~\eqref{equation:Hu} to $\Sph$ with normal vector $\u$, or its reflection $-H_{\u} = H_{-\u}$, contains at least one data point from $\Y$. Equivalently, $\u \in \dirs$ if and only if $H_{\u}$ contains at least one point from $\Y^*$ from Definition~\ref{definition:MTH}. Consequently, it is enough to consider only directions from $B(\u;\epsilon)$ from~\eqref{equation:spherical ball} for all $\u \in \dirs$ and any $\epsilon>0$ when computing \textSHD{}. In other words, to evaluate the infimum of $h(\u)$ from \eqref{equation:objective function} over $\u \in \Sph$, it suffices to consider some neighborhoods of all the points $\u\in\dirs$. We obtain our first simplification of \eqref{equation:scatter depth} in the form
    \begin{equation}    \label{computation formula}
    \SHD\left(\Y \right) = \inf_{\u\in\dirs} f_{\epsilon}(\u) \quad \mbox{ for any }\epsilon >0,
    \end{equation}
where 
    \begin{equation}    \label{fv}
    f_{\epsilon}(\u) = \inf_{\vv\in B(\u;\epsilon)} h(\vv)
    \end{equation}
denotes the infimum of the considered masses over all directions from the neighborhood $B(\u;\epsilon)$. 

\subsection{Spherical shells and (anti-)circles}   \label{section:shells}

For $\u \in \Sph$ we need to determine for each $i\in\{1,\dots,n\}$ whether $\y_i$ belongs to sets $\slab{\u}$ and $\cslab{\u}$, respectively. The set $\cslab{\u}$ is a union of a closed halfspace $H^+_{\u}=\left\{\bm x \in\R^d\colon \u\tr \bm x \ge 1\right\}$ and its reflection $H_{\u}^- = -H^+_{\u}=\left\{\bm x \in\R^d\colon \u\tr \bm x \le -1\right\}$ around the origin.

If $\left\Vert \y_i \right\Vert \leq 1$, certainly $\y_i \in \slab{\u}$ for all $\u \in \Sph$, as follows for example by the Cauchy-Schwarz inequality $|\u\tr \y_i| \leq \left\Vert \u \right\Vert \left\Vert \y_i \right\Vert \leq 1$. Consider therefore the case $\left\Vert \y_i \right\Vert \geq 1$. We have that $\y_i\in H^+_{\u}$ is equivalent with $\u\tr \y_i \ge 1$. Therefore, the set of those $\u \in \Sph$ such that $\y_i \in H^+_{\u}$ is the intersection of the unit sphere $\Sph$ with the halfspace $\left\{ \x \in \R^d \colon \x\tr \y_i \geq 1 \right\}$, whose boundary passes through the point $\y_i/\left\Vert \y_i \right\Vert^2$ and is orthogonal to the vector $\y_i$. That is a cap of the unit sphere; we denote it by
    \begin{equation}\label{equation:cap}
    S_i = \Sph \cap \left\{ \x \in \R^d \colon \x\tr \y_i \geq 1 \right\}. 
    \end{equation}  
Analogously, $\y_i \in H^-_{\u}$ if and only if $\left\Vert \y_i \right\Vert \geq 1$ and $\u \in -S_i$. Denoting further by
    \begin{equation}\label{equation:segment}
    C_i = \Sph \cap \left\{ \x \in \R^d \colon \left\vert \x\tr \y_i \right\vert \leq 1 \right\} 
    \end{equation}
the closed spherical segment of $\Sph$ complementary to $S_i \cup (-S_i)$, we have that     
    \begin{equation}
    \begin{aligned}
    \y_i \in \slab{\u} \mbox{ if and only if }\u \in C_i, \\
    \y_i \in \cslab{\u} \mbox{ if and only if }\u \in S_i \cup (- S_i).
    \end{aligned}
    \end{equation}
The sets $S_i$, $-S_i$, and $C_i$ cover the unit sphere. A point $\y_i$ lies in the boundary of the slab $\slab{\u}$, that is in both sets $\slab{\u}$ and $\cslab{\u}$ at the same time, if and only if $\u$ is in the spherical boundary of the set $C_i$, see Figures~\ref{fig:d2} and~\ref{fig:d3}. 

The triples of sets $S_i$, $-S_i$, and $C_i$, for all $i$ such that $\left\Vert \y_i \right\Vert > 1$, segment the sphere $\Sph$. For $d=2$ this segmentation amounts to a simple partition of the unit $1$-sphere $\Sph[1]$ into circular intervals, as can be observed in Figure~\ref{fig:d2}. For $d=3$, the spherical boundary of $S_i$ is a $1$-sphere (a circle) on the unit sphere $\Sph[2]$, see Figure~\ref{fig:d3}. For the general situation $d\geq 2$ we obtain that the spherical boundary of $S_i$ is the intersection of a hyperplane with $\Sph$ in $\R^d$. That set can be written as the $(d-2)$-sphere (a $0$-sphere or two points in $\R^2$, a $1$-sphere in $\R^3$, or a $2$-sphere in $\R^4$) centered at the point $\y_i/\left\Vert \y_i \right\Vert^2$ with radius $\sqrt{1 - \left\Vert \y_i \right\Vert^{-2}}$,
lying on the surface of the unit sphere $\Sph$. In what follows it will be useful to identify a point $\y_i$ such that $\left\Vert \y_i \right\Vert > 1$ with the spherical boundary of the corresponding cap $S_i$. In~\eqref{equation:circle} in Section~\ref{section:main theorem}, the spherical boundary of $S_i$ was called the \emph{circle} of $\y_i$, and the spherical boundary of $-S_i$ the \emph{anti-circle} of $\y_i$. Each (anti-)circle is thus a $(d-2)$-sphere lying in the unit sphere $\Sph$.

Considering the whole (transformed) random sample $\Y$ we see that the circles and anti-circles (or, equivalently the sets $S_i$, $-S_i$, $C_i$, for $i=1,\dots,n$ such that $\left\Vert \y_i \right\Vert > 1$) segment the unit sphere $\Sph$ into a finite number of spherical polytopes. We call these regions spherical \emph{shells}.

\begin{remark}
A direction $\vv\in\Sph$ belongs to $\dirs$ if and only if it lies on any (anti-)circle of $\y_i$ for some $i\in\{1,\dots,n\}$. Consequently, the spherical boundary of every shell is composed of points from $\dirs$.
\end{remark}

\subsection{Exploring the shells: Vertices of the shells and tilting} \label{section:tilting}

Under our assumption of general position\footnote{General position of $\X$ w.r.t. the ellipsoid~\eqref{equation:Mahalanobis} is equivalent with general position of $\Y$ w.r.t. $\Sph$.} of $\Y$ w.r.t. $\Sph$ from Definition~\ref{definition:general position}, the value of the function $f_{\epsilon}$ from~\eqref{computation formula} simplifies substantially. Using this simplification, only finitely many directions from $\dirs$ need to be considered to obtain the exact \textSHD{}.

\begin{lemma}   \label{lemma:tilting}
For any $\u \in \Sph$ and $\epsilon > 0$ small enough we can write function $f_{\epsilon}$ from~\eqref{fv} as
    \begin{equation}\label{comp}
    f_{\epsilon}(\u) = \min\left\{ \#\left\{ i \in \{ 1, \dots, n\} \colon \left\vert \u\tr \y_i \right\vert < 1 \right\}, \#\left\{ i \in \{ 1, \dots, n\} \colon \left\vert \u\tr \y_i \right\vert > 1 \right\} \right\}.  
    \end{equation}
In particular, there exists $\epsilon_0 > 0$ small enough so that for any $\epsilon \in (0,\epsilon_0)$ the right hand side of~\eqref{comp} does not depend on $\epsilon$. Let $\fdirs$ be a finite subset of $\dirs$ such that for any shell $S$ of $\Sph$ there exists an element of $\fdirs$ in the spherical boundary of $S$. Then we can write
    \begin{equation}\label{equation:computation formula 2} 
    \begin{aligned}
    \SHD\left( \Y \right) & = \min_{\u \in \fdirs} f_{\epsilon}(\u) \\
    & = \min_{\u \in \fdirs} \min\left\{ \#\left\{ i \in \{ 1, \dots, n\} \colon \left\vert \u \tr \y_i \right\vert < 1 \right\}, \#\left\{ i \in \{ 1, \dots, n\} \colon \left\vert \u \tr \y_i \right\vert > 1 \right\} \right\}.
    \end{aligned}    
    \end{equation}
\end{lemma}

\begin{proof}
Take any direction $\u \in \Sph$ and consider those data points from $\Y$ that lie in the interior of the slab $\slab{\u}$. Since there are only finitely many of them, there must exist $\epsilon>0$ such that for any direction $\vv \in \Sph$ that belongs to the neighborhood $B(\u;\epsilon)$ of $\u$ from \eqref{equation:spherical ball}, the interior of $\slab{\vv}$ contains the same data points as the interior of $\slab{\u}$. In addition, using the same argument, we can select $\epsilon>0$ such that also for all $\vv\in B(\u;\epsilon)$ all the data points outside $\slab{\u}$ coincide with the data points outside $\slab{\vv}$. Finally, we also take $\epsilon < \sqrt{2}$ which guarantees that for all $\vv \in B(\u;\epsilon)$ we have $\vv\tr \u > 0$. 

Take now $\vv \in B(\u;\epsilon)$ and examine what happens with the points on the boundary of $\slab{\u}$ if we tilt $\u$ to $\vv \ne \u$. The intersection 
    \[  H_{\u} \cap \slab{\vv} = \left\{ \x \in \R^d \colon \x\tr \u = 1 \mbox{ and }\left\vert \x\tr \vv \right\vert \leq 1 \right\}  \]
is a $(d-1)$-dimensional relatively closed\footnote{A set is \emph{relatively closed} if it is closed when considered as a subset of its affine hull. Analogously, the \emph{relative boundary} of $A$ is the boundary of $A$ when considered as a subset of its affine hull.} slab inside the hyperplane $H_{\u}$. Its $(d-2)$-dimensional boundaries, when considered inside the hyperplane $H_{\u}$, are orthogonal to the projection of the vector $\vv - \u$ into $H_{\u}$. The unit normal of the $(d-1)$-dimensional slab $H_{\u} \cap \slab{\vv}$ with $\vv \ne \u$ is given by the vector
    \[  \n(\u,\vv) = \frac{\vv - \left( \vv\tr \u \right) \u}{\left\Vert \vv - \left( \vv\tr \u \right) \u \right\Vert} = \frac{\vv - \left( \vv\tr \u \right) \u}{\sqrt{1 -  \left(\vv\tr \u \right)^2}} \in \Sph \]
parallel with $H_{\u}$. For a visualization of our situation with $d=3$ see Figure~\ref{figure:tilt}. As the distance between $\u$ and $\vv$ decreases, the inner product $\vv\tr\u$ tends to $1$, and the slab $H_{\u} \cap \slab{\vv}$ becomes wider in the direction $-\n(\u,\vv)$. Indeed, we have that for any $\x \in H_{\u} \cap \slab{\vv}$ we can bound
    \[  \frac{-1 - \vv\tr\u}{\sqrt{1 - \left(\vv\tr \u \right)^2}} \leq \x\tr\left(\n(\u,\vv)\right) = \frac{\x\tr\vv - \left(\vv\tr\u\right) \left(\x\tr\u\right)}{\sqrt{1 - \left(\vv\tr \u \right)^2}} \leq \frac{1 - \vv\tr\u}{\sqrt{1 - \left(\vv\tr \u \right)^2}},  \]
and as $\vv\tr\u \to 1$ we have that the left hand side diverges to $-\infty$, while the right hand side approaches $0$ from the right. Therefore, as $\vv_n \to \u$, the intersection $H_{\u} \cap \slab{\vv_n}$ forms a sequence of slabs \begin{enumerate*}[label=(\roman*)] \item containing $\u$, \item with boundaries orthogonal to $\n(\u,\vv_n)$, \item whose one boundary approaches the point $\u$ and the other tends to infinity, \end{enumerate*} see Figure~\ref{figure:tilt}. Because both $\slab{\vv}$ and $H_{\u} \cup H_{-\u}$ are symmetric around the origin, a data point $\y_i$ lies inside $H_{-\u} \cap \slab{\vv}$ if and only if its reflection $-\y_i$ lies inside $H_{\u} \cap \slab{\vv}$. 

\begin{figure}
    \centering
    \includegraphics[width=.45\textwidth]{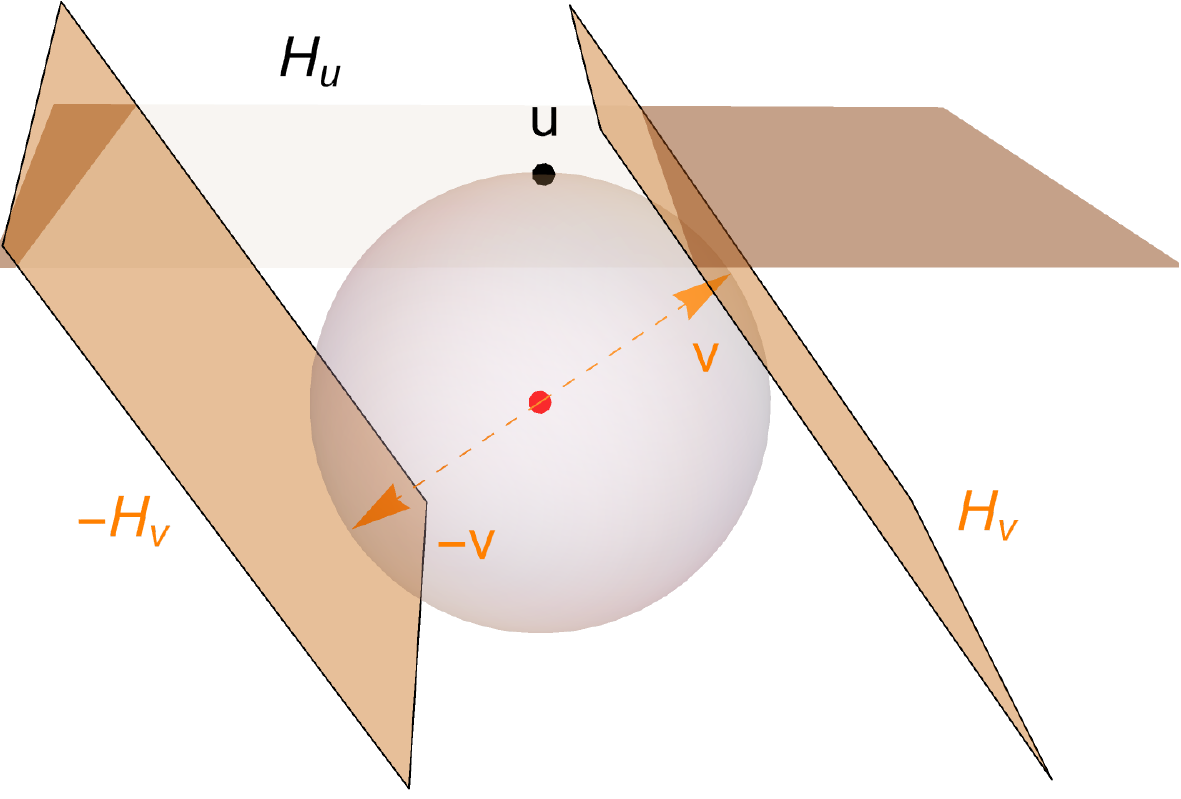} \quad \includegraphics[width=.40\textwidth]{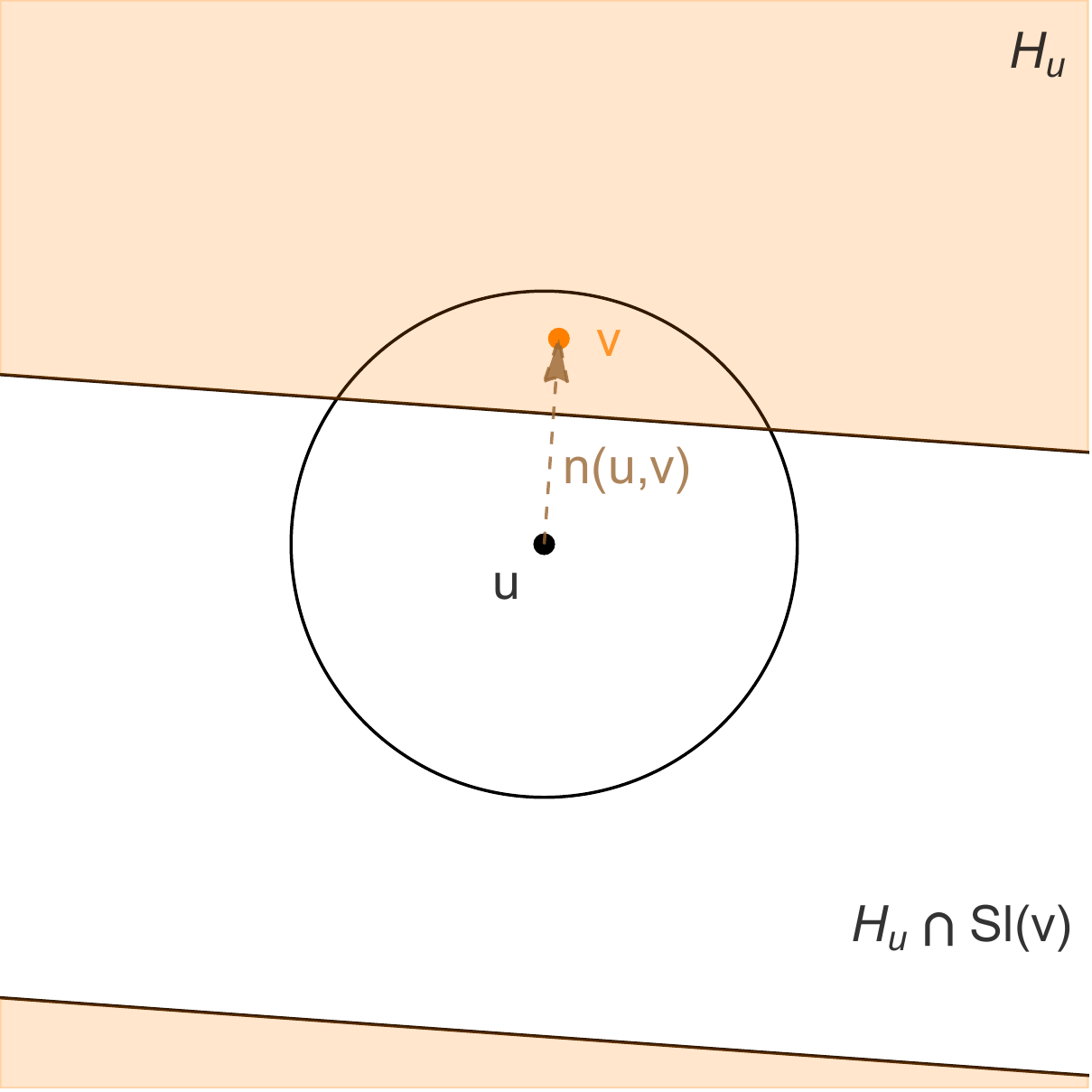}
    \caption{Proof of Lemma~\ref{lemma:tilting}: The effect of a slight tilt of the slab $\slab{\u}$ to $\slab{\vv}$ for $d=3$. In the left hand panel we see the intersection of the boundary plane $H_{\u}$ with $\slab{\vv}$. The region of intersection is a slab in the two-dimensional space $H_{\u}$ (right hand figure), whose boundary is orthogonal to $\n(\u,\vv)$, the projection of the difference $\vv - \u$ into $H_{\u}$. In the right hand panel the points indicated as $\u$ and $\vv$ represent the orthogonal projections of $\u$ and $\vv$ into $H_{\u}$, respectively. As $\vv$ approaches $\u$, the slab $H_{\u} \cap \slab{\vv}$ has one of the boundaries approaching $\u$, and the other one passing to infinity in the direction $-\n(\u,\vv)$. In the limit, we obtain a halfspace in $H_{\u}$ whose boundary passes through $\u$.}
    \label{figure:tilt}
\end{figure}

Because of our assumption from Definition~\ref{definition:general position}, we know that among the points $\y_1, \dots, \y_n, - \y_1$, $\dots, - \y_n$ from $\Y^*$, there are at most $d-1$ points inside $H_{\u}$, and these points together with $\u$ lie in general position. In particular, there exists a closed $(d-1)$-dimensional halfspace inside the hyperplane $H_{\u}$ with $\u$ in its relative boundary that contains no point from $\Y^*$. Denote by $\w \in \Sph$ the outer unit normal of that halfspace. The direction $\w$ is parallel with $H_{\u}$. We select two particular directions:
    \begin{itemize}
        \item Choosing $\vv = (\u + t\,\w)/\lVert\u + t\,\w\rVert$ for $t > 0$ small enough, we can guarantee that inside the slab $\slab{\vv}$ we find precisely all data points $\y_i$ that lie in the interior of $\slab{\u}$ and no points from the boundary of $\slab{\u}$.
        \item Taking $\vv = (\u - t\,\w)/\lVert\u - t\,\w\rVert$ for $t > 0$ small enough, we have that all points from the interior and the boundary of $\slab{\u}$ are inside the interior of $\slab{\vv}$, meaning that the complementary slab $\cslab{\vv}$ contains exactly those points that are not contained in $\slab{\u}$.
    \end{itemize} 
Summarizing our observations, we conclude that the minimum of the number of points from $\Y \cap \slab{\vv}$ and the number of points from $\Y \cap \cslab{\vv}$ over $\vv\in B(\u;\epsilon)$ (or equivalently the minimum of $h(\vv)$ over $\vv\in B(\u;\epsilon)$), for any $\epsilon>0$ small enough, is equal to the minimal value of the number of data points from $\Y$ that are contained in the interior of the slab $\slab{\u}$ and the number of data points from $\Y$ that are outside $\slab{\u}$. We have shown~\eqref{comp}.

As for our second claim, from our construction we see that this minimum is attained at some $\vv\in B(\u;\epsilon)\setminus\{ \u\}$ if $H_\u$ contains some points from $\Y^*$ (or equivalently, if $\u\in\dirs$), because we have shown that there exists a tilt of $\u$ to $\vv$ such that the points from the boundary of $\slab{\u}$ can be safely ignored in the computation of $h(\vv)$. That means that the infimum in $\inf_{\vv \in B(\u;\epsilon)} h(\vv)$ is always attained at some direction that belongs to the spherical interior of some shell. As the value of $h(\vv)$ is constant inside the interior of any shell, it is enough to pick a single direction $\u$ from the spherical boundary of every shell in formula~\eqref{computation formula}, instead of considering all the directions from $\dirs$. This brings about our second substantial simplification of the expression for the sample scatter halfspace depth --- instead of~\eqref{equation:scatter depth} or~\eqref{computation formula}, it can be written as~\eqref{equation:computation formula 2}. 
\end{proof}

In Lemma~\ref{lemma:tilting} we found that the minimization in \textSHD{} reduces to finding a minimum of a quite simple expression~\eqref{comp} over any finite set of directions $\fdirs \subset \dirs$ such that each spherical shell generated by $\Y$ contains at least one $\u \in \fdirs$ on its spherical boundary. It remains to construct a set $\fdirs$ with this property. We find $\fdirs$ among the maximal directions from Definition~\ref{definition:MTH}. Before doing so in Section~\ref{section:maximal tangent hyperplanes}, an auxiliary lemma will be useful.

\subsection{Dimensionality of intersections of (anti-)circles}   \label{section:intersections}

Under our condition $\Y$ in general position w.r.t. $\Sph$, intersections of (anti-)circles are always ``non-degenerate'' spheres in an appropriate lower-dimensional affine subspace. 

\begin{lemma}   \label{lemma:intersection}
The intersection of $k \in \left\{2, \dots, d-1 \right\}$ distinct (anti-)circles is either an empty set, or a $(d-k-1)$-sphere inside $\Sph$. 
\end{lemma}

\begin{proof}
Denote by $\a_1, \dots, \a_{k}$ the data points from $\Y^*$ that correspond to the (anti-)circles, respectively. Suppose that the intersection of the (anti-)circles is non-empty, and take $\u \in \Sph$ in this intersection. By the definition of an (anti-)circle we then have that $H_{\u}$ contains all points $\a_1, \dots, \a_{k}$, and by our assumption of general position w.r.t. the unit sphere (Definition~\ref{definition:general position}) we thus know that the points $\a_1, \dots, \a_k$ lie in general position inside $H_{\u}$. Since $k \leq d-1$, these points must be affinely independent, which is equivalent with the general position of the hyperplanes\footnote{A system of $k \in \left\{1, \dots, d\right\}$ hyperplanes in $\R^d$ is said to be in general position if any collection of $k$ of them intersects in an $(d-k)$-dimensional affine space, $k\leq d$.} $H_j = \left\{ \x \in \R^d \colon \x\tr \a_j = 1 \right\}$. The last claim can be seen, e.g., using duality considerations for convex polytopes \cite[Section~2.4]{Schneider2014}. We also know by~\eqref{equation:cap} that the $j$-th (anti-)circle is the intersection of $H_j$ with $\Sph$. The intersection of all $k$ (anti-)circles can therefore be expressed as $\Sph \cap \left( \bigcap_{j=1}^k H_j \right)$, where $\bigcap_{j=1}^k H_j$ is a $(d-k)$-dimensional affine subspace of $\R^d$, as we wanted to show.
%
%
\end{proof}

\subsection{Finding representative points: Maximal tangent hyperplanes}    \label{section:maximal tangent hyperplanes}

Recall from Definition~\ref{definition:MTH} that $H_{\u}$ is a maximal tangent hyperplane if $\u\in\dirs$ and if there is no direction $\vv\in\Sph$ such that $H_{\u}\cap \Y^*$ is a strict subset of $H_{\vv}\cap \Y^*$. It is instructive to see this definition also in terms of the unit normal $\u \in \Sph$.  Let $\u \in \dirs$ lie in an intersection $C$ of $k$ (anti-)circles determined by points from $\Y$. This direction is maximal if and only if there is no additional $(k+1)$-st (anti-)circle that intersects $C$. We intend to show now that in~\eqref{equation:computation formula 2} it suffices to consider only unit normal vectors $\fdirs$ of all maximal tangent hyperplanes. 

We concluded in~\eqref{equation:computation formula 2} that we need to pick at least one point $\u$ from the spherical boundary of every shell. The following lemma shows that the collection of all directions $\u \in \dirs$ that correspond to maximal tangent hyperplanes has this property. 

\begin{lemma}   \label{lemma:MTH}
Let $S \subset \Sph$ be a spherical shell. Then there exists a maximal direction $\u \in S \cap \dirs$. 
\end{lemma}

\begin{proof}
The shell $S$ is defined as an intersection of a finite number of spherical caps $S_i$, $-S_i$ from~\eqref{equation:cap}, and spherical segments $C_i$ from~\eqref{equation:segment}, $i\in \left\{ 1,\dots,n \right\}$. Each cap and segment can be represented as an intersection of the unit sphere $\Sph$ with one (for cap) or two (for segment) halfspaces in $\R^d$. In particular, the shell $S$ can be written in the form
    \[  S = \Sph \cap \left( \bigcap_{j=1}^J H_j \right)    \]
for some $J \geq 1$ and $H_1, \dots, H_J$ a collection of halfspaces in $\R^d$.

We start from any direction $\vv \in S \cap \dirs$ in the spherical boundary of $S$, meaning that $\vv$ must be contained in the boundary of some of the halfspaces $H_1, \dots, H_J$. Without loss of generality, suppose that
    \begin{equation}\label{equation:vv}
    \vv \in \Sph \cap \left( \bigcap_{j=1}^K \bd{H_j} \right) \cap \left( \bigcap_{j=K+1}^J \intr{H_j} \right)  
    \end{equation}
for $1 \leq K \leq J$. The $K$ elements $\Sph \cap \bd{H_j}$, $j=1,\dots,K$, correspond exactly to the (anti-)circles in $\Sph$ in which $\vv$ is contained; denote
    \begin{equation}\label{equation:C}
    C = \Sph \cap \left( \bigcap_{j=1}^K \bd{H_j} \right).  
    \end{equation}
If $\vv$ is not a maximal direction, there must exist some $(K+1)$-st (anti-)circle that intersects $C$ at a point $\u \in C$, $\u \ne \vv$. We show that a point $\u$ with this property can be found also in the spherical boundary of the shell $S$. Starting from any $\vv$ from the spherical boundary of $S$ that is not a maximal direction we can therefore get another direction $\u$ from the spherical boundary of $S$ that belongs to one more (anti-)circle than $\vv$. Iterating this search with $\vv$ replaced by $\u$ until no additional (anti-)circle intersects the set $C$ from~\eqref{equation:C}, we make sure that $\u$ from the last iteration corresponds to a maximal tangent hyperplane. That concludes our proof.

It remains to find a point in $S$ with the property above. If $\u \in S$, our assertion is trivially true. If $\u \notin S$, then there must exist at least one halfspace among $H_{K+1}, \dots, H_J$ in which $\u$ is not contained; let $\ell \in \left\{ K+1, \dots, J \right\}$ be the first index such that $\u \notin H_{\ell}$. The set $A = \bigcap_{j=1}^K \bd{H_j}$ is an affine subspace of $\R^d$. Since $\vv, \u \in C = \Sph \cap A$, the space $A$ has non-empty intersection $C$ with $\Sph$. In particular, as an intersection of an affine subspace and the unit sphere, $C$ must be a sphere inside the space $A$, and thus it is connected. Take any circular arc connecting the points $\vv$ and $\u$ inside $C$. Since we know by~\eqref{equation:vv} that $\vv \in \intr{H_{\ell}}$ and $\u \notin H_{\ell}$, there must be a point $\u'$ on this circular arc lying on the boundary of $H_{\ell}$. We have found a direction $\u' \in C \cap \bd{H_{\ell}}$. Now, if $\u' \in S$ we are done since we have $\u' \in S \cap \dirs$ that lies in the intersection of both $C$ and the additional (anti-)circle $\Sph \cap \bd{H_{\ell}}$. If $\u' \notin S$, we iterate our procedure with $\u$ replaced by $\u'$. This time we know that our new choice of $\u = \u'$ lies inside $C$ defined by~\eqref{equation:C} and the (anti-)circle given by $\bd{H_{\ell}}$. But, since $\u \notin S$ it must be that, there is another halfspace from $H_{\ell+1}, \dots, H_{J}$ not containing $\u$. Let $\ell' \in \left\{ \ell + 1, \dots,  J \right\}$ be the first index such that $\u\notin H_{\ell'}$. We iterate our procedure until we exhaust all indices $\ell$ such that $\u \notin H_{\ell}$, and finally find a point $\u \in S$ that lies in the intersection of $C$ and an additional (anti-)circle as required. 
\end{proof}

Lemma~\ref{lemma:MTH} reduces the directions used in~\eqref{computation formula} from all (anti-)circles $\u \in \dirs$ to its subset of maximal directions. That lemma, however, applies also to the finer result~\eqref{equation:computation formula 2} with finitely many elements in the minimum. Suppose that $H_{\u}$ is a maximal tangent hyperplane and denote $\{\a_1,\dots, \a_k\}=H_{\u}\cap \Y^*$ the maximal subset of $k$ data points of $\Y$ contained in the boundary of the slab $\slab{\u}$. Due to the general position of $\Y$ w.r.t. $\Sph$ from Definition~\ref{definition:general position}, we know that $\a_1,\dots,\a_k$ must be in general position and also $k\leq d-1$. We distinguish two cases:
    \begin{enumerate}[align=left]
    \item[\textbf{Case $k=d-1$.}] By Lemma~\ref{lemma:intersection}, the (non-empty) intersection of all $d-1$ (anti-)circles is a $0$-sphere inside a one-dimensional line intersecting $\Sph$. In other words, there are exactly two intersection directions $\u$ and $\u'$ in $\Sph$ corresponding to two distinct maximal tangent hyperplanes $H_{\u}$ and $H_{\u'}$.
    \item[\textbf{Case $k<d-1$.}] Applying Lemma~\ref{lemma:intersection} again, we have that the intersection of these $k$ (anti-)circles is a $(d - k - 1)$-sphere $C$ inside $\Sph$, with $d - k - 1 \geq 1$. Because of the definition of a maximal tangent hyperplane, no additional (anti-)circle intersects $C$. For any direction $\u \in C$ we therefore have that the slabs $\slab{\u}$ contain exactly the same points from $\Y$ in their interior, boundary, and complement, respectively. For any $\u\in C$ it thus holds that the function $f_{\epsilon}$ from~\eqref{comp} is constant over $\u \in C$ if $\epsilon > 0$ is small enough. Thus, in formulas~\eqref{computation formula} and~\eqref{equation:computation formula 2} it is enough to pick an arbitrary single direction $\u \in C$ and plug it into~\eqref{computation formula} as a representative of $C \subset \dirs$.
    \end{enumerate}

Combining the results obtained throughout this section we complete the proof of Theorem~\ref{theorem:main}.

\section{Proof of Lemma~\ref{lemma:intersection1}}  \label{section:proof intersection1}

There exists a direction $\u \in \Sph$ in the intersection of the $k$ circles if and only if the tangent hyperplane $H_{\u}$ of $\Sph$ that passes through $\u$ contains all the $k$ data points corresponding to the circles. That follows directly from the definition of a circle. The hyperplane $H_{\u}$ then necessarily contains the affine hull $A$ of these points, and because $H_{\u}$ does not intersect $B$, neither does the affine hull $A$. For the other implication, should the affine hull $A$ intersect the open ball $B$, then necessarily any hyperplane that contains $A$ intersects $B$, and thus it cannot be tangent to $\Sph$.

\section{Proof of Lemma~\ref{lemma:2}}  \label{section:proof lemma2}

Because $\bm 0_d \in C$, certainly $A \subset A'$, and if $A \cap B \ni \a$, then also $\a \in A' \cap B$. On the other hand, suppose that $A \cap B = \emptyset$, meaning that $\left\Vert \a \right\Vert \geq 1$ for all $\a \in A$. Then for any $\a \in A$ and $\c \in C$ we have by the orthogonality of all $\z_j^*$, $j=1,\dots, k$, with $\c$ that $\left\Vert \a + \c \right\Vert^2 = \left\Vert \a \right\Vert^2 + \left\Vert \c \right\Vert^2 \geq \left\Vert \a \right\Vert^2 > 1$. This inequality follows using Pythagoras' theorem because $\a$, as a linear combination of $\{\z_j^*\}_{j=1}^k$, is orthogonal to the vector $\c$. Necessarily, no point of $A'$ can intersect $B$ as we wanted to show.  

\vskip 0.1 in
\section*{Acknowledgments}
\paragraph{}
\vskip 0.1 in

The research is supported by NSF of China (Grant No.11971208, 11601197), China Postdoctoral Science Foundation funded project (2016M600511, 2017T100475), NSF of Jiangxi Province (No.2018ACB21002, 20171ACB21030), Student research project of Jiangxi University of Finance and Economics (No. 20200613104157285). P.~Laketa was supported by the OP RDE project ``International mobility of research, technical and administrative staff at the Charles University", grant CZ.02.2.69/0.0/0.0/18\_053/0016976. The work of S.~Nagy was supported by Czech Science Foundation (EXPRO project n. 19-28231X). 


\begin{thebibliography}{}

\bibitem[Chen et~al., 2018]{Chen_etal2018}
Chen, M., Gao, C., and Ren, Z. (2018).
\newblock Robust covariance and scatter matrix estimation under {H}uber's
  contamination model.
\newblock {\em Ann. Statist.}, 46(5):1932--1960.

\bibitem[Donoho and Gasko, 1992]{Donoho_Gasko1992}
Donoho, D.~L. and Gasko, M. (1992).
\newblock Breakdown properties of location estimates based on halfspace depth
  and projected outlyingness.
\newblock {\em Ann. Statist.}, 20(4):1803--1827.

\bibitem[Dyckerhoff, 2004]{Dyckerhoff2004}
Dyckerhoff, R. (2004).
\newblock Data depths satisfying the projection property.
\newblock {\em Allg. Stat. Arch.}, 88(2):163--190.

\bibitem[Dyckerhoff and Mozharovskyi, 2016]{Dyckerhoff_Mozharovskyi2016}
Dyckerhoff, R. and Mozharovskyi, P. (2016).
\newblock Exact computation of the halfspace depth.
\newblock {\em Comput. Statist. Data Anal.}, 98:19--30.

\bibitem[Dyckerhoff et~al., 2021]{Dyckerhoff_etal2021}
Dyckerhoff, R., Mozharovskyi, P., and Nagy, S. (2021).
\newblock Approximate computation of projection depths.
\newblock {\em Comput. Statist. Data Anal.}, 157:107166.

\bibitem[Eddelbuettel, 2013]{Eddelbuettel2013}
Eddelbuettel, D. (2013).
\newblock {\em Seamless {R} and {C++} Integration with {R}cpp}.
\newblock Use {R}! Springer, New York.

\bibitem[Eddelbuettel and Sanderson, 2014]{Eddelbuettel_Sanderson2014}
Eddelbuettel, D. and Sanderson, C. (2014).
\newblock Rcpp{A}rmadillo: {A}ccelerating {R} with high-performance {C}++
  linear algebra.
\newblock {\em Comput. Statist. Data Anal.}, 71:1054--1063.

\bibitem[Edelsbrunner, 1987]{Edelsbrunner1987}
Edelsbrunner, H. (1987).
\newblock {\em Algorithms in combinatorial geometry}, volume~10 of {\em EATCS
  Monographs on Theoretical Computer Science}.
\newblock Springer-Verlag, Berlin.

\bibitem[Horn and Johnson, 2013]{Horn_Johnson2013}
Horn, R.~A. and Johnson, C.~R. (2013).
\newblock {\em Matrix analysis}.
\newblock Cambridge University Press, Cambridge, second edition.

\bibitem[Liu et~al., 2019]{Liu_etal2019}
Liu, X., Mosler, K., and Mozharovskyi, P. (2019).
\newblock Fast computation of {T}ukey trimmed regions and median in dimension
  {$p>2$}.
\newblock {\em J. Comput. Graph. Statist.}, 28(3):682--697.

\bibitem[Liu and Zuo, 2014]{Liu_Zuo2014}
Liu, X. and Zuo, Y. (2014).
\newblock Computing halfspace depth and regression depth.
\newblock {\em Comm. Statist. Simulation Comput.}, 43(5):969--985.

\bibitem[Nagy, 2020]{Nagy2020}
Nagy, S. (2020).
\newblock Scatter halfspace depth: geometric insights.
\newblock {\em Appl. Math.}, 65(3):287--298.

\bibitem[Nagy et~al., 2020]{Nagy_etal2020}
Nagy, S., Dyckerhoff, R., and Mozharovskyi, P. (2020).
\newblock Uniform convergence rates for the approximated halfspace and
  projection depth.
\newblock {\em Electron. J. Statist.}, 14(2):3939--3975.

\bibitem[Paindaveine and Van~Bever, 2018]{Paindaveine_VanBever2018}
Paindaveine, D. and Van~Bever, G. (2018).
\newblock Halfspace depths for scatter, concentration and shape matrices.
\newblock {\em Ann. Statist.}, 46(6B):3276--3307.

\bibitem[Rousseeuw and Ruts, 1996]{Rousseeuw_Ruts1996}
Rousseeuw, P.~J. and Ruts, I. (1996).
\newblock {Algorithm {AS} 307: Bivariate location depth.}
\newblock {\em J. R. Stat. Soc., Ser. C}, 45(4):516--526.

\bibitem[Rousseeuw et~al., 1999]{Rousseeuw_etal1999}
Rousseeuw, P.~J., Ruts, I., and Tukey, J.~W. (1999).
\newblock The bagplot: A bivariate boxplot.
\newblock {\em Am. Stat.}, 53(4):382--387.

\bibitem[Rousseeuw and Struyf, 1998]{Rousseeuw_Struyf1998}
Rousseeuw, P.~J. and Struyf, A. (1998).
\newblock Computing location depth and regression depth in higher dimensions.
\newblock {\em Stat. Comput.}, 8(3):193--203.

\bibitem[Schneider, 2014]{Schneider2014}
Schneider, R. (2014).
\newblock {\em Convex bodies: the {B}runn-{M}inkowski theory}, volume 151 of
  {\em Encyclopedia of Mathematics and its Applications}.
\newblock Cambridge University Press, Cambridge, expanded edition.

\bibitem[Tukey, 1975]{Tukey1975}
Tukey, J.~W. (1975).
\newblock Mathematics and the picturing of data.
\newblock In {\em Proceedings of the {I}nternational {C}ongress of
  {M}athematicians ({V}ancouver, {B}. {C}., 1974), {V}ol. 2}, pages 523--531.
  Canad. Math. Congress, Montreal, Que.

\bibitem[Zhang, 2002]{Zhang2002}
Zhang, J. (2002).
\newblock Some extensions of {T}ukey's depth function.
\newblock {\em J. Multivariate Anal.}, 82(1):134--165.

\end{thebibliography}

\def\cprime{$'$} \def\polhk#1{\setbox0=\hbox{#1}{\ooalign{\hidewidth
  \lower1.5ex\hbox{`}\hidewidth\crcr\unhbox0}}}

\end{document}